\documentclass[preprint,11pt]{elsarticle}

\usepackage{ucs}
\usepackage[latin1]{inputenc} 
\usepackage[T1]{fontenc}
\usepackage{latexsym,amssymb,amsthm,bbm,amsxtra,amstext}
\usepackage{palatino}
\usepackage{pstricks,pst-node}
\usepackage{amsfonts,epsfig,multicol,anysize}
\usepackage{multicol}
\usepackage{float}
\usepackage{longtable}
\usepackage{calrsfs}
\usepackage{multirow}
\usepackage{booktabs}
\usepackage{threeparttable}
\usepackage{caption}
\usepackage{afterpage}
\usepackage[colorlinks]{hyperref}

\newtheorem{theorem}{Theorem}[section]
\newtheorem{definition}{Definition}[section]
\newtheorem{corollary}{Corollary}[theorem]

\newtheorem{algorithm}{Algorithm}[section]

\providecommand{\n}[1]{{\boldsymbol{#1}}}
\def\diag{\mathop{\rm diag}\nolimits}
\newcommand{\e}{\mbox{E}\ell} 
\newcommand{\mb}{\mbox{BCE}\ell} 
\newcommand{\mbn}{\mbox{BCN}} 
\newcommand{\mbt}{{\mbox{BC}}{\it{t}}} 
\newcommand{\te}{\mbox{TE}\ell} 
\newcommand{\bcs}{\mbox{BCS}} 
\newcommand{\lel}{\mbox{LE}\ell} 
\newcommand{\mln}{\mbox{LN}} 
\newcommand{\mlt}{\mbox{L}{\it{t}}} 

\usepackage{amssymb}

\usepackage{lineno}

\journal{}

\begin{document}

\begin{frontmatter}



\title{Box--Cox elliptical distributions with application}



\author[moran]{Raúl Alejandro Morán-Vásquez}
\author[ferrari]{Silvia L. P. Ferrari}
\address[moran]{Institute of Mathematics, University of Antioquia, Medellín, Colombia}
\address[ferrari]{Department of Statistics, University of São Paulo, São Paulo, Brazil}

\begin{abstract}
We propose and study the class of Box--Cox elliptical distributions. It provides alternative distributions for modeling 
multivariate positive, marginally skewed and possibly heavy-tailed data. This new class of distributions has as a 
special case the class of log-elliptical distributions, and reduces to the Box--Cox symmetric class of distributions in the 
univariate setting. The parameters are interpretable in terms of quantiles and relative dispersions of the marginal 
distributions and of associations between pairs of variables. The relation between the scale parameters and quantiles makes the
Box--Cox elliptical distributions attractive for regression modeling purposes. Applications to data on vitamin intake are presented and discussed. 
\end{abstract}


\end{frontmatter}


\section{Introduction}
\label{S:1}

Multivariate positive data are frequently found in empirical studies. The statistical analysis of such data often relies 
on the multivariate normal distribution assumptions, ignoring characteristics of the data, namely the positive support and
possible skewness and presence of outlying observations. Improvements for accommodating outliers may be achieved by replacing 
the multivariate normal distribution by a heavy-tailed distribution in the elliptical class of distributions, such as the 
multivariate $t$ distribution (Lange \textit{et al}. \cite{LANGE1}). Futher improvement may be achieved by the use of 
log-skew-elliptical distributions (Marchenko and Genton \cite{MAR1}), which are multivariate distributions with support in 
$\mathbb{R}_{+}^p$ and accomodate heavy-tailed distributions. An alternative methodology for modeling multivariate positive 
data uses a Box--Cox transformation in each component of the vector of observations. In this approach one assumes that 
the vector of transformed observations follows a multivariate normal or an elliptical distribution 
(Quiroz \textit{et al}. \cite{QUIROZ}). This assumption implies a theoretical shortcoming because the support of the 
transformed vector of observations is not necessarily $\mathbb{R}^p$. Moreover, in this approach the model parameters are 
interpretable only in terms of characteristics of the transformed observations (not the original variables of interest). 
In the univariate case, Ferrari and Fumes \cite{FUMES1} overcome these shortcomings by proposing the class of Box--Cox 
symmetric distributions. This class includes several alternative distributions, such as the Box--Cox Cole--Green 
(Stasinopoulos \textit{et al}. \cite{STAS1}), Box--Cox $t$ (Rigby and Stasinopoulos \cite{RIG-STAS3}), Box--Cox power 
exponential (Rigby and Stasinopoulos \cite{RIG-STAS1}, Voudouris \textit{et al}. \cite{VOU}) distributions, and a new 
distribution, the Box--Cox slash distribution, for modeling univariate positive, skewed, possibly heavy-tailed data.

In the present paper, we focus on the problem of constructing a class of multivariate distributions with support in 
$\mathbb{R}_{+}^p$ in such a way that the marginal distributions have properties similar to those of the Box--Cox 
symmetric distributions, the parameters are interpretable and association among variables is controlled by 
association parameters. We name the proposed class of distributions as the Box--Cox elliptical class of distributions. It
has the log-elliptical class of distributions (Fang \textit{et al}. \cite{FANG1}) as a special subclass and reduces to the 
Box--Cox symmetric class of distributions in the univariate setting. The construction of the new class is performed through 
an extension of the Box--Cox transformation and involves another new class of distributions defined in  this paper, the class of 
truncated elliptical distributions. The parameters of the Box--Cox elliptical distributions are interpretable as characteristics 
of the original variables (not the transformed variables). Some parameters are related to quantiles of the marginal distributions, 
which makes the Box--Cox elliptical distributions attractive for regression modeling purposes. Several properties of the proposed distributions are derived. In particular, some properties of the log-elliptical distributions that are not available in the 
literature are direct consequences of properties of Box--Cox elliptical distributions stated in this paper. The flexibility 
of the proposed distributions for modeling multivariate positive, asymmetric data with or without the presence of outlying 
observations is illustrated through an analysis of real data on vitamin intake by older people.

The paper is organized as follows. In Section \ref{S:2} we define the truncated elliptical distributions and present some 
properties. In Section \ref{S:3} we define the family of the extended Box--Cox transformations, we use it to define the class 
of Box--Cox elliptical distributions, and we state several properties. In Section \ref{S:4} we give interpretation for the 
parameters and show the relation between some parameters and quantiles of the marginal distributions. In Section \ref{S:5} 
we focus on maximum likelihood estimation and present simulation studies. In Section \ref{S:6} we present applications to 
real data. Finally, Section \ref{S:7} closes the paper with concluding remarks. Technical proofs are presented in the Appendix.

\section{The class of the truncated elliptical distributions}
\label{S:2}

In this section, we define a new class of distributions named the class of the truncated elliptical distributions. 
Its is needed for the definition and study of the class of the Box--Cox elliptical distributions, which is the focus of this paper. 

We denote vectors and their components with lowercase Greek letters in bold and normal fonts, respectively. For instance, 
if $\n{\xi}\in{\mathbb{R}}^p$, then $\n{\xi}=(\xi_1,\ldots,\xi_p)'$. Additionally, 
$\n{\xi}_{-k}\in\mathbb{R}^{p-1}$, $k=1,\ldots,p$, is the sub-vector obtained from $\n{\xi}$ by excluding its $k$-th component. 
Similar notations are used for random vectors, but we use capital Roman letters. Matrices are denoted by capital Greek letters 
in boldface and their entries in lowercase normal font Greek letters. For example, if $\n{\Delta}(p\times q)$ is a matrix 
with components in ${\mathbb{R}}$, then $\n{\Delta}=(\delta_{jk})_{p\times q}$. If $\n{\Delta}$ is a symmetric matrix, the 
notation $\n{\Delta}>0$ means that $\n{\Delta}$ is positive definite. If $\n{\Delta}(p\times p)>0$, then 
$\n{\Delta}_{-k,k}\in\mathbb{R}^{p-1}$ is the sub-vector obtained by deleting the $k$-th component of the $k$-th column of 
$\n{\Delta}$; $\n{\Delta}_{k,-k}=\n{\Delta}_{-k,k}'$; and $\n{\Delta}_{-k,-k}>0$ is the sub-matrix obtained by excluding the 
$k$-th row and the $k$-th column of $\n{\Delta}$.

The elliptical distributions have been extensively studied in the statistical literature and applied in different fields; 
see Fang \textit{et al}. \cite{FANG1}, Gupta \textit{et al}. \cite{GUP1} and references therein. From now on, whenever we say that 
a random vector has an elliptical distribution we assume that its probability density function (PDF) exists. 

\begin{definition}\label{elip}
The random vector $\n{X}\in{\mathbb{R}}^p$ has an elliptical distribution with location vector $\n{\mu}\in{\mathbb{R}}^p$ and
dispersion matrix $\n{\Sigma}(p\times p)>0$, if its PDF is
\begin{equation}\label{dens-eliptica1}
f_{\n{X}}(\n{x})=c_p\det(\n{\Sigma})^{-1/2}g((\n{x}-\n{\mu})'\n{\Sigma}^{-1}(\n{x}-\n{\mu})),\quad\n{x}\in\mathbb{R}^p.
\end{equation}
The function $g$, called density generating function (DGF), is such that $g(u)\geq0$, for all $u\geq0$, and 
$\int_{0}^{\infty}r^{p-1}g(r^2)\,\textrm{d}r<\infty$. The normalizing constant $c_p$ is 
$$c_p = \frac{\Gamma(p/2)}{2\pi^{p/2}} \left(\int_{0}^{\infty}r^{p-1}g(r^2)\,\textrm{d}r\right)^{-1}.$$
We write $\n{X}\sim\e_p(\n{\mu},\n{\Sigma};g)$. 
\end{definition}

The univariate case of Definition \ref{elip} corresponds to a random variable $X$ having a symmetric distribution with location 
parameter $\mu\in\mathbb{R}$, dispersion parameter $\sigma^2>0$ and \,DGF $g$, and we write $X\sim\e_1(\mu,\sigma^2;g)$. 
A detailed study about elliptical distributions can be found in Fang \textit{et al}. \cite{FANG1}.

\begin{definition}\label{defeliptrunc}
Let $B\subseteq{\mathbb{R}}^p$ be a measurable set. The random vector $\n{W}\in B$ has a truncated elliptical distribution with 
support $B$ and parameters $\n{\mu}\in\mathbb{R}^p$ and $\n{\Sigma}(p\times p)>0$, DGF $g$, and we write 
$\n{W}\sim\te_p(\n{\mu},\n{\Sigma};B;g)$, if its PDF is
\begin{equation}\label{eliptrun}
f_{\n{W}}(\n{w})=\frac{g((\n{w}-\n{\mu})'\n{\Sigma}^{-1}(\n{w}-\n{\mu}))}{\int_{B}g((\n{w}-\n{\mu})'\n{\Sigma}^{-1}(\n{w}-\n{\mu}))\,
{\rm{d}}\n{w}},\quad \n{w}\in B,
\end{equation}
where $g$ is such that $g(u)\geq0$, for all $u\geq0$, and $\int_{0}^{\infty}t^{p-1}g(t^2)\,\textrm{d}t<\infty$.
\end{definition}

If $B=\mathbb{R}^p$ in (\ref{eliptrun}), we arrive at PDF (\ref{dens-eliptica1}). 

The univariate case of Definition \ref{defeliptrunc} corresponds to a random variable, say $W$, with a truncated symmetric 
distribution with support $B\subseteq\mathbb{R}$, parameters $\mu\in\mathbb{R}$ \, and $\sigma^2>0$, DGF $g$, and we write 
$W\sim\te_1(\mu,\sigma^2;B;g)$. 

Each member of the class of the truncated elliptical distributions is characterized by the DGF $g$. Two notable special cases are the 
multivariate truncated normal and truncated $t$ distributions, which correspond to the DGF $g(u)\propto\exp(-u/2)$ and 
$g(u)\propto(1+u/\tau)^{-(\tau+p)/2}$, with $\tau>0$, respectively. Other special cases include the following multivariate distributions: 
truncated power exponential ($g(u)\propto\exp(-u^{\beta}/2)$, $\beta>0$), 
truncated slash ($g(u)\propto\int_{0}^{1}t^{p+q-1}\exp(-ut^2/2)\,{\rm d}t$, $q>0$), and 
truncated scale mixture of normal distributions ($g(u)\propto\int_{0}^{\infty}t^{p/2}\exp(-ut/2)\,{\rm d}H(t)$, $u\geq0$, $H$ being a
cumulative distribution function (CDF) on $(0,\infty)$). The DGF $g$ may include extra parameters in PDF (\ref{eliptrun}). For instance, 
the multivariate truncated $t$ distribution has the degrees of freedom parameter $\tau$, that controls the tail behaviour. The multivariate 
truncated normal distribution is a limiting case of the multivariate truncated $t$ distribution when $\tau\to\infty$. Some studies on 
multivariate truncated normal distributions are found in Birnbaum and Meyer \cite{BIR1}, Tallis \cite{TALL1,TALL2,TALL3}, Horrace \cite{HORR1} 
and Manjunath and Wilhelm \cite{MANJU1}. The multivariate truncated $t$ distribution with rectangular support is considered in 
Ho \textit{et al}.~\cite{HO1}. 

Let $W\sim\te_1(\mu,\sigma^2;(a,b);g)$. The CDF of $W$ is given by
\begin{equation}\label{fda-trunc-univ}
F_{W}(w)=\frac{F_{Z}\bigl(\frac{w-\mu}{\sigma}\bigr)-F_{Z}\bigl(\frac{a-\mu}{\sigma}\bigr)}{F_{Z}\bigl(\frac{b-\mu}{\sigma}\bigr)-
F_{Z}\bigl(\frac{a-\mu}{\sigma}\bigr)},\quad w\in(a,b),
\end{equation}
where $F_Z$ is the CDF of a random variable $Z$ having a standard symmetric distribution, $Z\sim\e_1(0,1;g)$. Equation 
(\ref{fda-trunc-univ}) is also valid when $a\to-\infty$ \,and/or \,$b\to\infty$. In this case, we have 
$F_{Z}((a-\mu)/\sigma)\to 0$ \,and/or\, $F_{Z}((b-\mu)/\sigma)\to1$.

Let $R=I_1\times\cdots\times I_p$ be a rectangle in $\mathbb{R}^p$, where $I_1,\ldots,I_p$ are intervals in $\mathbb{R}$ (finite or infinite). 
With no loss of generality, assume that $I_k=(a_k,b_k)$, $k=1,\ldots,p$. 
\begin{theorem}\label{tr-eltrun-condic}
If $\n{W}\sim\te_p(\n{\mu},\n{\Sigma};R;g)$, then
$W_k|\n{W}_{-k}\sim\te_1(\mu_{k.-k},\sigma_{k.-k}^2;(a_k,b_k);g_{k.-k})$, $k=1,\ldots,p$, where
$\mu_{k.-k} = \mu_k+\n{\Sigma}_{k,-k}\n{\Sigma}_{-k,-k}^{-1}(\n{w}_{-k}-\n{\mu}_{-k})$, 
$\sigma_{k.-k}^2 = \sigma_{kk} - \n{\Sigma}_{k,-k}\n{\Sigma}_{-k,-k}^{-1}\n{\Sigma}_{-k,k}$ \,and\, $g_{k.-k}(u) = g(u+q(\n{w}_{-k}))$, with
$q(\n{w}_{-k})=(\n{w}_{-k}-\n{\mu}_{-k})'\n{\Sigma}_{-k,-k}^{-1}(\n{w}_{-k}-\n{\mu}_{-k})$.
\end{theorem}
\begin{proof}
See \hyperlink{proof-tr-eltrun-condic}{Appendix \ref{proof-tr-eltrun-condic}}.
\end{proof}

Theorem \ref{tr-eltrun-condic} states that if a random vector $\n{W}$ has a truncated elliptical distribution with its support
being a rectangle in $\mathbb{R}^p$, then the conditional distribution of $W_k$ given $\n{W}_{-k}$ is truncated symmetric
with the same support of $W_k$. This fact is useful for obtaining the complete conditional distributions, from which
random samples from (\ref{fda-trunc-univ}) may be obtained using the inverse transformation method. This allows us to propose 
Algorithm \ref{algo-elip-trun} to generate random samples of the random vector $\n{W}\sim\te_p(\n{\mu},\n{\Sigma};R;g)$. 
We construct a Markov chain by sampling from the complete conditional distributions of $W_k|\n{W}_{-k}$, $k=1,\ldots,p$, 
given in Theorem \ref{tr-eltrun-condic}. Let $\n{w}^{(j)}$ be a sample generated in the $j$-th iteration, $j=1,\ldots,n$.

\begin{algorithm}\label{algo-elip-trun}
\begin{enumerate}
\item[\phantom{alejo}] 
\item Choose a starting value $\n{w}^{(0)}$ of the Markov chain.
\item Generate a random variable $u$ from a uniform distribution ${\mbox{U}}(0,1)$.
\item In each cycle $j=1,\ldots,n$, apply the inverse transformation method using (\ref{fda-trunc-univ}) to compute
\begin{equation*}
w_{k.-k}^{(j)}=\mu_{k.-k}^{(j)}+\sigma_{k.-k}^{(j)}F_{Z_k}^{-1}\biggl[u\biggl\{F_{Z_k}\biggl(\frac{b_k-\mu_{k.-k}^{(j)}}{\sigma_{k.-k}^{(j)}}\biggr)-
F_{Z_k}\biggl(\frac{a_k-\mu_{k.-k}^{(j)}}{\sigma_{k.-k}^{(j)}}\biggr)\biggr\}+F_{Z_k}\biggl(\frac{a_k-\mu_{k.-k}^{(j)}}{\sigma_{k.-k}^{(j)}}\biggr)\biggr],
\end{equation*}
where $Z_k\sim\e_1(0,1;g_{k.-k})$, for $k=1,\ldots,p$. This is the sampled value from the conditional distribution of
\begin{equation*}
w_k^{(j)}\,|\,w_1^{(j)},\ldots, w_{k-1}^{(j)},w_{k+1}^{(j-1)},\ldots,w_p^{(j-1)},\quad k=1,\ldots,p.
\end{equation*}
\end{enumerate}
\end{algorithm}

\section{The class of the Box--Cox elliptical distributions}
\label{S:3}

In this section, we define the class of the Box--Cox elliptical distributions and state several properties. First, we define the family of 
the extended Box--Cox transformations, which is a generalization of multivatiate Box--Cox transformations given in Quiroz \textit{et al}. 
\cite[Eq. $1.1$, $1.2$]{QUIROZ}. Using this new family of transformations, we define the class of the Box--Cox elliptical distributions. 
We then present various properties of these distributions regarding a characterization through truncated elliptical distributions with 
rectangular support, marginal and conditional distributions, independence, and mixed moments. Some of these properties will be needed 
for interpreting the parameters of the Box--Cox elliptical distributions (see Section \ref{S:4}). 

For each $\n{\xi}\in\mathbb{R}^p$, let $\n{D}_{\n{\xi}}$ be a diagonal matrix with diagonal elements $\n{\xi}$, \textit{i.e.}, 
$\n{D}_{\n{\xi}}=\diag\{\xi_1,$ $\ldots,\xi_p\}$. Let $R(\n{\xi})=I(\xi_1)\times\cdots\times I(\xi_p)$ be a rectangle in $\mathbb{R}^p$, 
where
\begin{equation}\label{conjuntos}
I(\xi_k)=
 \begin{cases}
  (-1/\xi_k,\infty), & \text{$\xi_k>0$},\\
  (-\infty,-1/\xi_k), & \text{$\xi_k<0$},\\
  (-\infty,\infty), & \text{$\xi_k=0$},
 \end{cases}
\end{equation}
for $k=1,\ldots,p$.

\begin{definition}\label{def-bc-transf-esten}
Let $\n{\lambda}\in\mathbb{R}^p$ and $\n{\mu}\in\mathbb{R}_{+}^p$. The extended Box--Cox transformation is defined by
$T_{\n{\lambda},\n{\mu}}:\mathbb{R}_{+}^p\to R(\n{\lambda})$ for the random vector $\n{Y}\in\mathbb{R}_+^p$ as 
$T_{\n{\lambda},\n{\mu}}(\n{Y})=\n{W}$, where $\n{W}$ is the $p$-dimensional vector with $k$-th element given by
\begin{equation}\label{bc-transf-esten}
W_k=
 \begin{cases}
  \dfrac{(Y_k/\mu_k)^{\lambda_k}-1}{\lambda_k}, & \text{$\lambda_k\neq0$},\vspace{0.2cm}\\
  \quad\!\log(Y_k/\mu_k), & \text{$\lambda_k=0$},
 \end{cases} 
\end{equation}
for $k=1,\ldots,p$.
\end{definition}

From Definition \ref{def-bc-transf-esten} we have that $\mu_k$ is a scale parameter for $Y_k$, for $k=1,\ldots,p$. If $\n{\mu}=\n{1}=(1,\ldots,1)'$ 
in (\ref{bc-transf-esten}) we obtain the multivariate Box--Cox transformation (Quiroz \textit{et al}. \cite[Eq. $1.1$, $1.2$]{QUIROZ}). 
Also, $T_{\n{\lambda},\n{\mu}}(\n{Y})\to T_{\n{0},\n{\mu}}(\n{Y})$ when $\n{\lambda}\to\n{0}=(0,\ldots,0)'$. Moreover, 
if $\n{\alpha}\in\mathbb{R}_+^p$, then $T_{\n{\lambda},\n{\mu}}(\n{D_\alpha}\n{Y})=T_{\n{\lambda},\n{D}_{\n{\alpha}}^{-1}\n{\mu}}(\n{Y})$. 
If $\n{\beta}\in R(\n{\lambda})$ and $\n{\gamma}=\n{1}+\n{D_{\lambda}\beta}$, then 
$\n{D}_{\n{\gamma}}^{-1}(T_{\n{\lambda},\n{\mu}}(\n{Y})-\n{\beta})=T_{\n{\lambda},\n{\delta}}(\n{Y})$, where
$\n{\delta}=T_{\n{\lambda},\n{\mu}}^{-1}(\n{\beta})$. 
These facts allow us to derive various properties of the Box--Cox elliptical distributions.


\begin{definition}\label{def-bcelip}
We say that the random vector $\n{Y}\in\mathbb{R}_{+}^p$ has a Box--Cox elliptical distribution with parameters
$\n{\mu}\in\mathbb{R}_{+}^p$, $\n{\lambda}\in\mathbb{R}^p$, 
$\n{\Sigma}(p\times p)>0$ and DGF $g$ if 
$T_{\n{\lambda},\n{\mu}}(\n{Y})\sim\te_p(\n{0},\n{\Sigma};R(\n{\lambda});g)$, and we write 
$\n{Y}\sim\mb_p(\n{\mu},\n{\lambda},\n{\Sigma};g)$. 
\end{definition}

Equivalently, $\n{W}\sim\te_p(\n{0},\n{\Sigma};R(\n{\lambda});g)$ if 
$T_{\n{\lambda},\n{\mu}}^{-1}(\n{W})\sim\mb_p(\n{\mu},\n{\lambda},\n{\Sigma};g)$. If  
$\n{\lambda}=\n{0}$ in Definition \ref{def-bcelip}, then $\n{Y}$ follows a log-elliptical distribution with parameters
$\n{\mu}\in\mathbb{R}_+^p$, $\n{\Sigma}(p\times p)>0$ and DGF $g$ (Fang \textit{et al}. \cite{FANG1}), and we write 
$\n{Y}\sim\lel_p(\n{\mu},\n{\Sigma};g)$. 

From Definition \ref{def-bcelip} we have that the PDF of 
$\n{Y}\sim\mb_p(\n{\mu},\n{\lambda},\n{\Sigma};g)$ is given by
\begin{equation}\label{boxcoxelipmulti}
f_{\n{Y}}({\n y})=\dfrac{g(\n{w}'\n{\Sigma}^{-1}\n{w})\prod_{k=1}^{p}\frac{y_k^{\lambda_k-1}}{\mu_k^{\lambda_k}}}{\int_{R(\n{\lambda})}g(\n{w}'\n{\Sigma}^{-1}\n{w})\,{\rm d}\n{w}},\quad\n{w}=T_{\n{\lambda},\n{\mu}}(\n{y}),\quad\n{y}\in\mathbb{R}_{+}^p.
\end{equation} 

The case $p=1$ in (\ref{boxcoxelipmulti}) corresponds to the PDF of a positive random variable $Y$ with a Box--Cox symmetric 
distribution with parameters $\mu>0$, $\sigma>0$, $\lambda\in\mathbb{R}$ and DGF $g$ (Ferrari and Fumes \cite{FUMES1}), denoted by
$Y\sim\bcs(\mu,\sigma,\lambda;g)$. From Definition \ref{def-bcelip}, it is clear that each member of the class of the truncated 
elliptical distributions has its corresponding member in the class of the Box--Cox elliptical distributions, which is identified 
by its DGF $g$. Hence, by replacing $g(u)\propto\exp(-u/2)$,  $u\geq0$, in (\ref{boxcoxelipmulti}) we obtain 
the PDF of a random vector $\n{Y}\in\mathbb{R}_+^p$ with a multivariate Box--Cox normal distribution with parameters 
$\n{\mu}\in\mathbb{R}_{+}^p$, $\n{\lambda}\in\mathbb{R}^p$ \,and \,$\n{\Sigma}(p\times p)>0$, denoted by 
$\n{Y}\sim\mbn_p(\n{\mu},\n{\lambda},\n{\Sigma})$. When $g(u)\propto(1+u/\tau)^{-(\tau+p)/2}$, $\tau>0$, $u\geq0$, in 
(\ref{boxcoxelipmulti}) we have the PDF of a random vector $\n{Y}\in\mathbb{R}_+^p$ with a multivariate Box--Cox $t$ distribution 
with parameters $\n{\mu}\in\mathbb{R}_{+}^p$, $\n{\lambda}\in\mathbb{R}^p$, $\n{\Sigma}(p\times p)>0$ \,and \,$\tau>0$ degrees of 
freedom, denoted by $\n{Y}\sim\mbt_p(\n{\mu},\n{\lambda},\n{\Sigma};\tau)$. In these cases, when $\n{\lambda}=\n{0}$, we get the 
PDF of $\n{Y}\in\mathbb{R}_+^p$ with multivariate log-normal and log-$t$ distributions, denoted by 
$\n{Y}\sim\mln_p(\n{\mu},\n{\Sigma})$ and $\n{Y}\sim\mlt_p(\n{\mu},\n{\Sigma};\tau)$, respectively. As expected, the multivariate 
Box--Cox normal distribution is a limiting case of the multivariate Box--Cox $t$ distribution as $\tau\to\infty$. Other members of the 
class of the Box--Cox elliptical distributions include the multivariate Box--Cox power exponential distribution, 
the multivariate Box--Cox slash distribution, and the multivariate Box-Cox scale mixture of normal distributions.

Figure \ref{fig-bct-biv} shows plots of the PDF of 
$\n{Y}\sim\mbt_2(\n{\mu},\n{\lambda},\n{\Sigma};\tau)$ for different parameter values. The legend indicates the values of
all the parameters considered in the first plot and the value of the parameter that is changed from a plot to the next (in 
alphabetical order). Note that the parameter $\sigma_{12}$ impacts the association between the marginal distributions 
of $Y_1$ and $Y_2$ (Figures \ref{fig-bct-biv}(a), \ref{fig-bct-biv}(b) and \ref{fig-bct-biv}(c)). The
parameter $\mu_1$  affects the scale of the marginal distribution of $Y_1$ (Figures \ref{fig-bct-biv}(c) and 
\ref{fig-bct-biv}(d)). The parameter $\sigma_{22}$ influences the dispersion of the marginal distribution of 
$Y_2$ (Figures \ref{fig-bct-biv}(d) and \ref{fig-bct-biv}(e)). The parameters $\lambda_1$ and $\lambda_2$ control 
the skewness of the respective marginal distribution of $Y_1$ and $Y_2$ (Figures \ref{fig-bct-biv}(e), 
\ref{fig-bct-biv}(f), \ref{fig-bct-biv}(g) and \ref{fig-bct-biv}(h)). In Figure \ref{fig-bct-biv}(g), for which 
$\lambda_1=\lambda_2=1$, it is clear that the contour lines are (truncated) ellipses;
this fact is stated in item $3$ of Theorem \ref{prop-mbce}. Additionally, as the degrees of freedom parameter grows, the 
contour lines corresponding to the bivariate Box--Cox $t$ distributions tend to the contour lines of bivariate Box--Cox 
normal distributions. Moreover, the tails of the Box--Cox $t$ distributions seem to be heavier for smaller values of 
$\tau$ (Figures \ref{fig-bct-biv}(h) and \ref{fig-bct-biv}(i)).

Definition \ref{def-bcelip} characterizes the Box--Cox elliptical distributions from truncated elliptical distributions 
with support in $R(\n{\lambda})$ and parameter $\n{\mu}=\n{0}.$ 
In Theorem \ref{carac-bc-elip}, we present a characterization of Box--Cox elliptical distributions from truncated elliptical distributions with support in $R(\n{\lambda})$ and parameter $\n{\mu}=\n{\xi}$.

\begin{theorem}\label{carac-bc-elip}
Let $\n{\mu}\in\mathbb{R}_+^p$, $\n{\lambda}\in\mathbb{R}^p$, $\n{\xi}\in R(\n{\lambda})$, 
$\n{\alpha}=\n{1}+\n{D_{\lambda}\xi}$ \,and \,$\n{\Sigma}(p\times p)>0$. Then, 
$T_{\n{\lambda},\n{\mu}}(\n{Y})\sim\te_p(\n{\xi},\n{\Sigma};R(\n{\lambda});g)$ if and only if
$\n{Y}\sim\mb_p(\n{\delta},\n{\lambda},\n{D}_{\n{\alpha}}^{-1}\n{\Sigma}\n{D}_{\n{\alpha}}^{-1};g)$, 
where $\n{\delta}=T_{\n{\lambda},\n{\mu}}^{-1}(\n{\xi})$.
\end{theorem}
\begin{proof}
See Appendix \ref{proof-carac-bc-elip}.
\end{proof}

\begin{figure}
\captionsetup{font=scriptsize}
\centering
\includegraphics[scale=0.64,trim={0.45cm 0.7cm 0cm 0.3cm},clip]{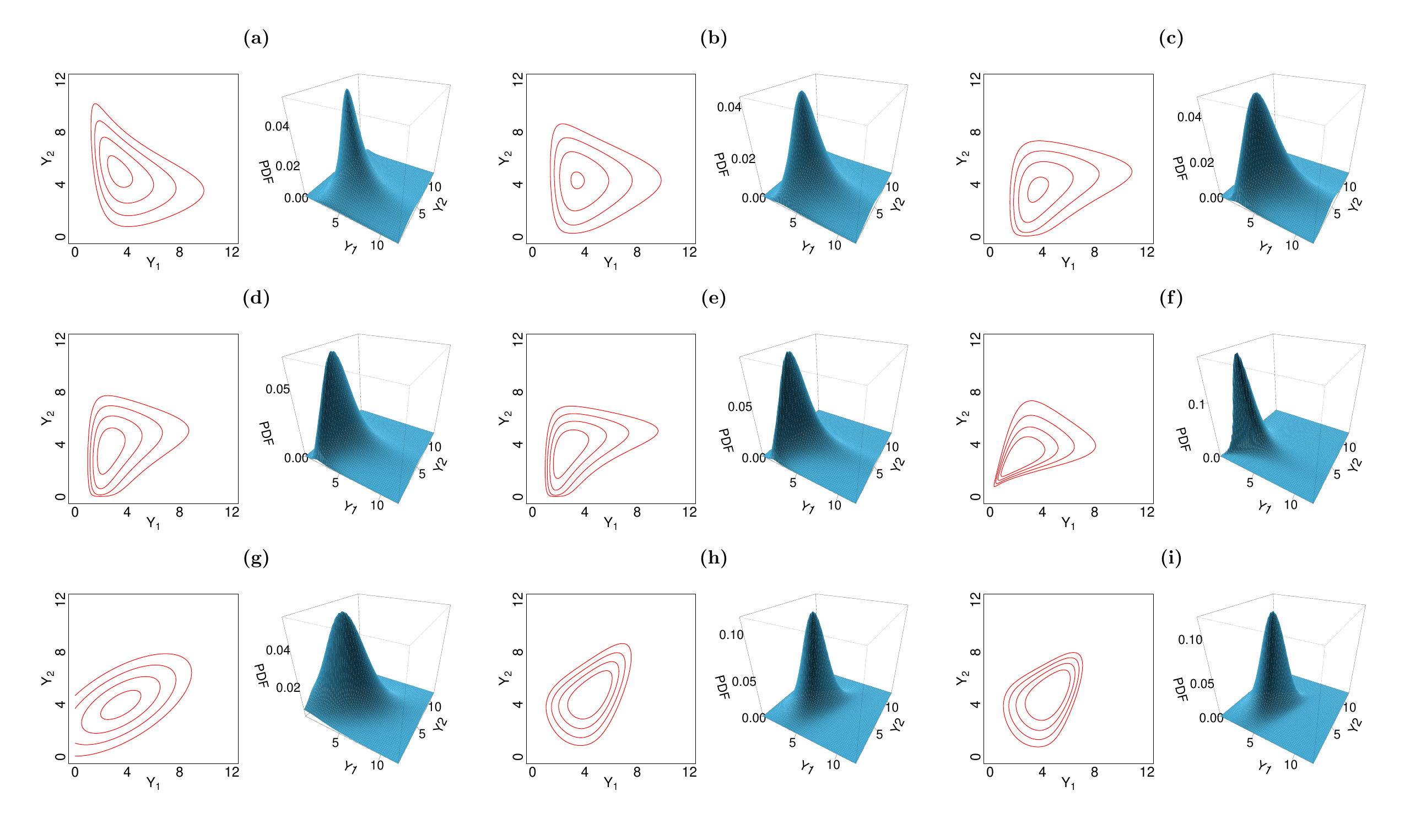}
\caption{Contour plots at levels $0.04$, $0.02$, $0.01$, $0.005$ and joint PDF of 
$\n{Y}\sim\mbt_2(\n{\mu},\n{\lambda},\n{\Sigma};\tau)$, where  \textbf{(a)} $\mu_1=5$, $\mu_2=4$, $\lambda_1=-1$, 
$\lambda_2=1.5$, $\sigma_{11}=0.5$, $\sigma_{22}=0.3$, $\sigma_{12}=-0.2$, $\tau=3$, \textbf{(b)} $\sigma_{12}=0$,
\textbf{(c)} $\sigma_{12}=0.2$, \textbf{(d)} $\mu_1=3.5$, \textbf{(e)} $\sigma_{22}=0.2$, \textbf{(f)} $\lambda_2=-1.5$,
\textbf{(g)} $\lambda_1=\lambda_2=1$, \textbf{(h)} $\lambda_2=2$, \textbf{(i)} $\tau=10$.}
\label{fig-bct-biv}
\end{figure}

In Theorem \ref{prop-mbce} we state various distributional results concerning the Box--Cox elliptical distributions. 
Items $1$ and $2$ consider some transformations of Box--Cox elliptical random vectors, and item $3$ states that the 
class of the truncated elliptical distributions with support in $\mathbb{R}_+^p$ and parameter $\n{\mu}\in\mathbb{R}_+^p$ 
is obtained from the class of the Box--Cox elliptical distributions.

\begin{theorem}\label{prop-mbce}
Let $\n{\mu}\in\mathbb{R}_+^p$, $\n{\lambda}\in\mathbb{R}^p$, 
$\n{\Sigma}(p\times p)>0$ \,and \,$\n{Y}\sim\mb_p(\n{\mu},\n{\lambda},\n{\Sigma};g)$.
\begin{enumerate}
 \item If $\n{\alpha}\in\mathbb{R}_{+}^p$, then $\n{D_{\alpha}Y}\sim\mb_p(\n{D_\alpha\mu},\n{\lambda},\n{\Sigma};g)$.
 \item If $\n{\beta}\in\mathbb{R}^p\setminus\n{0}$ and $\n{U}\in\mathbb{R}_+^p$ is the random vector 
 with components $U_k=(Y_k/\mu_k)^{\beta_k}$, $k=1,\ldots,p$, then  
 $\n{U}\sim\mb_p\bigl(\n{1},\n{D}_{\n{\beta}}^{-1}\n{\lambda},\n{D_\beta\Sigma D_\beta};g\bigr)$.
 \item If $\n{\lambda}=\n{1}$, then $\n{Y}\sim\te_p(\n{\mu},\n{D_{\mu}\Sigma D_{\mu}};\mathbb{R}_{+}^p;g)$.
\end{enumerate}
\end{theorem}
\begin{proof}
See Appendix \ref{proof-prop-mbce}.
\end{proof}

In order to state results on marginal and conditional distributions let us consider the partitions of
$\n{Y}\in\mathbb{R}_+^p$, $\n{\mu}\in\mathbb{R}_+^p$,
$\n{\lambda}\in\mathbb{R}^p$\, and \,$\n{\Sigma}(p\times p)>0$ as
\begin{equation}\label{particion-bc}
\n{Y}=(\n{Y}_1',\n{Y}_2')',\quad
\n{\mu}=(\n{\mu}_1',\n{\mu}_2')',\quad
\n{\lambda}=(\n{\lambda}_1',\n{\lambda}_2')',\quad
\n{\Sigma}=\begin{bmatrix}
             \n{\Sigma}_{11} & \n{\Sigma}_{12}\\
             \n{\Sigma}_{21} & \n{\Sigma}_{22}
            \end{bmatrix},
\end{equation}
with $\n{Y}_1\in\mathbb{R}_+^r$, $\n{Y}_2\in\mathbb{R}_+^{p-r}$, $\n{\mu}_1\in\mathbb{R}_+^r$, 
$\n{\mu}_2\in\mathbb{R}_+^{p-r}$, $\n{\lambda}_1\in\mathbb{R}^r$, $\n{\lambda}_2\in\mathbb{R}^{p-r}$, 
$\n{\Sigma}_{11}(r\times r)>0$, $\n{\Sigma}_{22}((p-r)\times(p-r))>0$ \,and \,$\n{\Sigma}_{12}(r\times(p-r))$ such that 
$\n{\Sigma}_{21}=\n{\Sigma}_{12}'$. The rectangle $R(\n{\lambda})$ can be written as
$R(\n{\lambda})=R(\n{\lambda}_1)\times R(\n{\lambda}_2)$, where 
$R(\n{\lambda}_1)=I(\lambda_1)\times\cdots\times I(\lambda_r)\in\mathbb{R}^r$ \,and \,$R(\n{\lambda}_2)=I(\lambda_{r+1})\times\cdots\times I(\lambda_p)\in\mathbb{R}^{p-r}$.

Let $\n{Y}\in\mathbb{R}_+^p$, $\n{\mu}\in\mathbb{R}_+^p$, $\n{\lambda}\in\mathbb{R}^p$, $\n{\Sigma}(p\times p)>0$ 
partitioned as in (\ref{particion-bc}) and such that $\n{Y}\sim\mb_p(\n{\mu},\n{\lambda},\n{\Sigma};g)$. The marginal PDF 
of $\n{Y}_1$ is given by
\begin{equation}
f_{\n{Y}_1}(\n{y}_1) = \dfrac{\bigr\{\int_{R(\n{\lambda_2})}g(\n{w}'\n{\Sigma}^{-1}\n{w})\,{\rm d}\n{w}_2\bigl\}\prod_{k=1}^{r}\frac{y_k^{\lambda_k-1}}{\mu_k^{\lambda_k}}}{\int_{R(\n{\lambda})}g(\n{w}'\n{\Sigma}^{-1}\n{w})\,{\rm d}\n{w}},\quad\n{y}_1\in\mathbb{R}_{+}^r,\label{margy1}\\
\end{equation}
where $\n{w}=(\n{w}_1',\n{w}_2')'$, with 
$\n{w}_1=T_{\n{\lambda}_1,\n{\mu}_1}(\n{y}_1)$ \,and\, $\n{w}_2=T_{\n{\lambda}_2,\n{\mu}_2}(\n{y}_2)$. Clearly, the  
marginal PDF (\ref{margy1}) is not necessarily of the form (\ref{boxcoxelipmulti}). This form is possible when 
$\n{\Sigma}_{12}=\n{0}$, \textit{i.e.}, when the matrix $\n{\Sigma}(p\times p)>0$ is block-diagonal. In Theorem \ref{bc-marg} 
this fact is stated.

\begin{theorem}\label{bc-marg}
Let $\n{Y}\in\mathbb{R}_+^p$, $\n{\mu}\in\mathbb{R}_+^p$, $\n{\lambda}\in\mathbb{R}^p$, $\n{\Sigma}(p\times p)>0$ 
partitioned as in (\ref{particion-bc}) and such that $\n{Y}\sim\mb_p(\n{\mu},\n{\lambda},\n{\Sigma};g)$. If 
$\n{\Sigma}_{12}=\n{0}$, then $\n{Y}_1\sim\mb_r(\n{\mu}_1,\n{\lambda}_1,\n{\Sigma}_{11};g_1)$, where
\begin{equation*}
g_1(u)=\int_{T(R(\n{\lambda}_2))}g(u+\n{s}'\n{s})\,{\rm{d}}\n{s},\quad u\geq0, 
\end{equation*}
with $T:\mathbb{R}^{p-r}\to\mathbb{R}^{p-r}$ being the transformation 
$T(\n{x})=\n{\Sigma}_{22}^{-1/2}\n{x}$.
\end{theorem}
\begin{proof}
See Appendix \ref{proof-bc-marg}.
\end{proof}

When the matrix $\n{\Sigma}(p\times p)>0$ is a diagonal matrix, all the marginal distributions are Box--Cox symmetric distributions.  
This fact is stated in Corollary \ref{marg-uni-bcs}.

\begin{corollary}\label{marg-uni-bcs}
Let $\n{\mu}\in\mathbb{R}_+^p$, $\n{\lambda}\in\mathbb{R}^p$, 
$\n{\Sigma}=\diag\{\sigma_{11},\ldots,\sigma_{pp}\}>0$ \,and \,$\n{Y}\sim\mb_p(\n{\mu},\n{\lambda},\n{\Sigma};g)$. Then, 
$Y_k\sim\bcs(\mu_k,\sqrt{\sigma_{kk}},\lambda_k;g_k)$, $k=1,\ldots,p$, where
\begin{equation*}\label{bc-marg-univ}
g_{k}(u)=\int_{R(\n{\Sigma}_{-k,-k}^{1/2}\n{\lambda}_{-k})}g(u+\n{s}'\n{s})\,{\rm{d}}\n{s}, \quad u\geq0.
\end{equation*}
\end{corollary}
\begin{proof}
Simply let $\n{Y}_1=Y_k$, $\n{Y}_2=\n{Y}_{-k}$, $\n{\mu}_1=\mu_k$, $\n{\mu}_2=\n{\mu}_{-k}$, 
$\n{\lambda}_1=\lambda_k$, $\n{\lambda}_2=\n{\lambda}_{-k}$, 
$\n{\Sigma}_{11}=\sigma_{kk}$ \,and \,$\n{\Sigma}_{22}=\n{\Sigma}_{-k,-k}$, $k=1,\ldots,p$, in Theorem \ref{bc-marg}.
\end{proof}

In Theorem \ref{bc-marg} we stated that if $\n{Y}=(\n{Y}_1',\n{Y}_2')'\sim\mb_p(\n{\mu},\n{\lambda},\n{\Sigma};g)$,
then the sub-vector $\n{Y}_1$ has a Box--Cox elliptical distribution if $\n{\Sigma}_{12}=\n{0}$.  Note that
$\n{Y}_1$ has a distribution in the Box--Cox elliptical class but not necessarily with the same parent distribution
as $\n{Y}$ (\textit{e.g.} normal, $t$, power exponential). The condition in Theorem \ref{bc-marg}, although sufficient, is not 
necessary for the subclass of the log-elliptical distributions. Indeed, if 
$\n{Y}=(\n{Y}_1',\n{Y}_2')'\sim\lel_p(\n{\mu},\n{\Sigma};g)$, then the sub-vector $\n{Y}_1$ has a 
log--elliptical distribution for any $\n{\Sigma}(p\times p)>0$ (Fang \textit{et al}. \cite[Sec. 2.8]{FANG1}). Moreover, the 
distribution of $\n{Y}_1$ is log-elliptical with the same parent distribution as $\n{Y}$ if the DGF $g$ is that of multivariate 
scale mixture of normal distributions, as we establish in Theorem \ref{logell-marg}.

\begin{theorem}\label{logell-marg}
Let $\n{Y}\in\mathbb{R}_+^p$, $\n{\mu}\in\mathbb{R}_+^p$, $\n{\Sigma}(p\times p)>0$ partitioned as in (\ref{particion-bc}) and such 
that $\n{Y}\sim\lel_p(\n{\mu},\n{\Sigma};g)$, with $g(u)\propto\int_{0}^{\infty}t^{p/2}\exp(-ut/2)\,{\rm d}H(t)$, $u\geq0$, $H$ 
being a CDF on $(0,\infty)$. Then, $\n{Y}_1\sim\lel_r(\n{\mu}_1,\n{\Sigma}_{11};g)$.
\end{theorem}
\begin{proof}
See Appendix \ref{proof-logell-marg}.
\end{proof}

The following log-elliptical distributions have DGF as multivariate scale mixture of normal distributions 
and therefore satisfy the conditions in Theorem \ref{logell-marg}: multivariate log--normal distribution 
($H$ is the CDF of a degenerate distribution at $t=1$), the multivariate log--$t$ distribution 
($H$ is the CDF of a gamma distribution with shape parameter $\tau/2$ and scale parameter $2/\tau$, $\tau>0$), 
the multivariate log--slash distribution ($H$ is the CDF of $T=U^{2/q}$, $q>0$, with $U\sim{\mbox{U}}(0,1)$), 
and the multivariate log--power exponential distribution for $0<\beta\leq1$ ($H$ is the CDF with PDF 
$h(t)=\frac{1}{2}t^{-3/2}h_{\beta}(t^{-1/2})$, $0<\beta<1$, with $h_{\beta}$ given in 
Gómez-Sánchez-Manzano \textit{et al.} \cite[Eq. 3]{GOM2}. If $\beta=1$, $H$ is as in the multivariate log--normal distribution
case). However, Theorem \ref{logell-marg} does not apply to the multivariate log--power exponential distribution for $\beta>1$.

In Theorem \ref{bc-cond} we state that, if $\n{Y}=(\n{Y}_1',\n{Y}_2')'$ has a Box--Cox elliptical distribution, then
the conditional distribution of $\n{Y}_1|\n{Y}_2$ is Box--Cox elliptical.

\begin{theorem}\label{bc-cond}
Let $\n{Y}\in\mathbb{R}_+^p$, $\n{\mu}\in\mathbb{R}_+^p$, $\n{\lambda}\in\mathbb{R}^p$, $\n{\Sigma}(p\times p)>0$ 
partitioned as in (\ref{particion-bc}) and such that $\n{Y}\sim\mb_p(\n{\mu},\n{\lambda},\n{\Sigma};g)$. Let
$\n{\mu}_1(\n{w}_2)=\n{\Sigma}_{12}\n{\Sigma}_{22}^{-1}\n{w}_2\in R(\n{\lambda}_1)$ \,and\, $\n{\alpha}(\n{w}_2)=\n{1}+\n{D}_{\n{\lambda}_1}\n{\mu}_1(\n{w}_2)$, 
with $\n{w}_2=T_{\n{\lambda}_2,\n{\mu}_2}(\n{y}_2)$, then 
$\n{Y}_1|\n{Y}_2=\n{y}_2\sim\mb_r(\n{\delta}_1,\n{\lambda}_1,\n{D}_{\n{\alpha}(\n{w}_2)}^{-1}\n{\Sigma}_{11\cdot2}\n{D}_{\n{\alpha}(\n{w}_2)}^{-1};g_{q(\n{w}_2)})$, 
where $\n{\delta}_1 = T_{\n{\lambda}_1,\n{\mu}_1}^{-1}(\n{\mu}_1(\n{w}_2))$, 
$\n{\Sigma}_{11\cdot2} = \n{\Sigma}_{11}-\n{\Sigma}_{12}\n{\Sigma}_{22}^{-1}\n{\Sigma}_{21}$ \,e\, $g_{q(\n{w}_2)}(u) = g(u+q(\n{w}_2))$, $u\geq0$, with $q(\n{w}_2) = \n{w}_2'\n{\Sigma}_{22}^{-1}\n{w}_2$.
\end{theorem}
\begin{proof}
See Appendix \ref{proof-bc-cond}. 
\end{proof}

If $\n{\Sigma}_{12}=\n{0}$ in Theorem \ref{bc-cond}, then
$\n{Y}_1|\n{Y}_2=\n{y}_2\sim\mb_r\bigl(\n{\mu}_1,\n{\lambda}_1,\n{\Sigma}_{11};g_{q(\n{w}_2)}\bigr)$. By comparing
this conditional distribution with the marginal distribution of $\n{Y}_1$ given in Theorem \ref{bc-marg}, we have that, if
$\n{\Sigma}_{12}=\n{0}$, $\n{Y}_1|\n{Y}_2$ \,and\, $\n{Y}_1$ have the same distribution if the DGFs
$g_{q(\n{w}_2)}$ and $g_1$ coincide. 
This characterizes the independence of the sub-vectors
$\n{Y}_1$ and $\n{Y}_2$, as we state in Theorem \ref{carac-bc-elip-norm}.

\begin{theorem}\label{carac-bc-elip-norm}
Let $\n{Y}\in\mathbb{R}_+^p$, $\n{\mu}\in\mathbb{R}_+^p$, $\n{\lambda}\in\mathbb{R}^p$, $\n{\Sigma}(p\times p)>0$ 
partitioned as in (\ref{particion-bc}) and such that $\n{Y}\sim\mb_p(\n{\mu},\n{\lambda},\n{\Sigma};g)$. Then, 
$\n{Y}_1$ and $\n{Y}_2$ are independent if and only if 
$\n{Y}\sim\mbn_p(\n{\mu},\n{\lambda},\n{\Sigma})$ \,and \, $\n{\Sigma}_{12}=\n{0}$.
\end{theorem}
\begin{proof}
See Appendix \ref{proof-carac-bc-elip-norm}. 
\end{proof}

In Theorem \ref{bc-momentos} we give an expression for mixed moments of Box--Cox elliptical random vectors.

\begin{theorem}\label{bc-momentos}
Let $\n{h}\in\mathbb{R}^p$, $\n{\mu}\in\mathbb{R}_+^p$, $\n{\lambda}\in\mathbb{R}^p$, $\n{\Sigma}(p\times p)>0$,
$\n{Y}\sim\mb_p(\n{\mu},\n{\lambda},\n{\Sigma};g)$ \,and\, $\n{U}\sim\mb_p(\n{1},\n{\lambda},\n{\Sigma};g)$. If 
${\mbox{E}}(\prod_{k=1}^{p}U_k^{h_k}) <\infty$, then
\begin{equation*}
{\mbox{E}}\biggl(\prod_{k=1}^{p}Y_k^{h_k}\biggr)=\biggl(\prod_{k=1}^{p}\mu_k^{h_k}\biggl){\mbox{E}}\biggl(\prod_{k=1}^{p}U_k^{h_k}\biggl).
\end{equation*}
\end{theorem}
\begin{proof}
See Appendix \ref{proof-bc-momentos}. 
\end{proof}

The computation of mixed moments of $\n{Y}\sim\mb_p(\n{\mu},\n{\lambda},\n{\Sigma};g)$ from mixed moments of
$\n{U}\sim\mb_p(\n{1},\n{\lambda},\n{\Sigma};g)$ as indicated in Theorem \ref{bc-momentos} is possible using Monte Carlo integration. 
By using Algorithm \ref{algo-elip-trun}, one may generate a random sample of size $n$ of the random vector
$\n{W}=T_{\n{\lambda},\n{1}}(\n{U})\sim{\mbox{TE}\ell}_p(\n{0},\n{\Sigma};R(\n{\lambda});g)$, say
$\n{w}_1,\ldots,\n{w}_n$, where $\n{w}_i=(w_{i1},\ldots,w_{ip})$, $i=1,\ldots,n$. If $n$ is large,
\begin{equation*}
{\mbox{E}}\biggl(\prod_{k=1}^{p}Y_k^{h_k}\biggr) \approx \frac{1}{n}\sum_{i=1}^{n}\prod_{k=1}^p \bigl(\mu_k u_k(w_{ik})\bigr)^{h_k},
\end{equation*}
with $u_k(w_{ik})=T_{\lambda_k,1}^{-1}(w_{ik})$, $i=1,\ldots,n$; $k=1,\ldots,p$. 

Let $\n{\lambda}=\n{0}$ (\textit{i.e.}, $\n{Y}\sim\lel_p(\n{\mu},\n{\Sigma};g)$) in Theorem \ref{bc-momentos}. We have that
\begin{equation*}
{\mbox{E}}\biggl(\prod_{k=1}^{p}Y_k^{h_k}\biggr)=\biggl(\prod_{k=1}^{p}\mu_k^{h_k}\biggl){\mbox{M}}_{\n{X}}(\n{h}),
\end{equation*}
whenever ${\mbox{M}}_{\n{X}}$, the moment generating function of $\n{X}\sim\e_p(\n{0},\n{\Sigma};g)$, exists 
(see Fang \textit{et al}. \cite[Sec. 2.8]{FANG1}).
Another consequence of Theorem \ref{bc-momentos} is that  the covariance matrices of
$\n{Y}\sim\mb_p(\n{\mu},\n{\lambda},\n{\Sigma};g)$ \,and\, $\n{U}\sim\mb_p(\n{1},\n{\lambda},\n{\Sigma};g)$, denoted by 
$\n{\Sigma}_{\n{Y}}$ and $\n{\Sigma}_{\n{U}}$, respectively, are such that $\n{\Sigma_Y}=\n{D_\mu}\n{\Sigma}_{\n{U}}\n{D_\mu}$.
Moreover, the correlation matrices of $\n{Y}$ and $\n{U}$ are equal. 

\section{Parameter interpretation}
\label{S:4}


From Definition \ref{def-bcelip} we have that the distribution of a random vector 
$\n{Y}\sim\mb_p(\n{\mu},\n{\lambda},\n{\Sigma};g)$ is characterized by a random vector
$\n{W}=T_{\n{\lambda},\n{\mu}}(\n{Y})\sim\te_p(\n{0},\n{\Sigma};R(\n{\lambda});g)$. In such a characterization, the parameter vectors 
$\n{\mu}\in\mathbb{R}_{+}^p$ \,and \,$\n{\lambda}\in\mathbb{R}^p$ are introduced through an extended Box--Cox transformation
(Definition \ref{def-bc-transf-esten}), in such a way that $\mu_k$ \,and \,$\lambda_k$, 
$k=1,\ldots,p$, are parameters involved in the transformation of $Y_k$ only; hence these parameters are
characteristics of the distribution of $Y_k$. Also, the marginal distributions of the components of $\n{W}$ are associated through 
$\n{\Sigma}(p\times p)>0$, which implies that the marginal distributions of the components of $\n{Y}$ are associated through this
matrix aswell. Hence, $\mu_k$ and $\lambda_k$, $k=1,\ldots,p$, are, respectively the scale parameter and skewness parameter 
(power transformation for marginal symmetry) of the distribution of $Y_k$; $\sigma_{jk}$, $j\neq k$, is the association parameter 
between $Y_j$ and $Y_k$.

The parameters $\mu_k$ and $\sigma_{kk}$, $k=1,\ldots,p$, are related with quantiles of $Y_k$. In order to establish these relations, 
let the marginal PDF of $Y_k$ be written as
\begin{equation}\label{margy_k}
f_{Y_k}(y_k) = \frac{g_{{\n{\Upsilon}_k}}(s_k)\frac{y_k^{\lambda_k-1}}{\sqrt{\sigma_{kk}}\mu_k^{\lambda_k}}}{\int_{I(\lambda_k\sqrt{\sigma_{kk}})}g_{{\n{\Upsilon}_k}}(s_k)\,{\rm d}s_k},\quad s_k=\sigma_{kk}^{-1/2}T_{\lambda_k,\mu_k}(y_k),\quad y_k>0, 
\end{equation}
with $I(\lambda_k\sqrt{\sigma_{kk}})$ defined in (\ref{conjuntos}) \,and
\begin{equation}\label{fdp_marg_aux}
g_{{\n{\Upsilon}_k}}(u_k) = \int_{R(\n{\Delta}_{-k,-k}\n{\lambda}_{-k})}g((1+\n{\Upsilon}_k\n{\Upsilon}_k')u_k^2 - 
2\n{\Upsilon}_k\n{\Omega}_ku_k\n{w} + \n{w}'\n{\Omega}_k'\n{\Omega}_k\n{w})\,{\rm{d}}\n{w},
\end{equation}
where $u_k\in I(\lambda_k\sqrt{\sigma_{kk}})$, $\n{\Delta}=\diag\{\sqrt{\sigma_{11}},\ldots,\sqrt{\sigma_{pp}}\}$, 
$\n{\Omega}_k=(\n{\Sigma}_{-k,-k} - \sigma_{kk}^{-1}\n{\Sigma}_{-k,k}\n{\Sigma}_{k,-k})^{-1/2}\n{\Delta}_{-k,-k}$ and
$\n{\Upsilon}_k=\sigma_{kk}^{-1/2}\n{\Sigma}_{k,-k}\n{\Omega}_k\n{\Delta}_{-k,-k}^{-1}$. 

PDF (\ref{margy_k}) can be built from a random variable $U_k$ defined in $\mathbb{R}$ with CDF
\begin{equation}\label{vble-aux-quan-fda}
F_{U_k}(u_k) = c_k\int_{-\infty}^{u_k} g_{{\n{\Upsilon}_k}}(t)\,{\rm d}t,
\end{equation}
where $c_k^{-1} = \int_{-\infty}^{\infty}g_{{\n{\Upsilon}_k}}(t)\,{\rm d}t$ (see details in Appendix \ref{proof-margy_k}). An interesting 
case occurs when the integral that involves $g_{{\n{\Upsilon}_k}}$ has integration region 
$R(\n{\Delta}_{-k,-k}\n{\lambda}_{-k})=\mathbb{R}^{p-1}$, \textit{i.e.}, when $\n{\Delta}_{-k,-k}\n{\lambda}_{-k}=\n{0}$. In this 
case, $U_k\sim\e_1(0,1;\tilde{g})$, with
\begin{equation}\label{aux-interp-bcel}
\tilde{g}(u) = \int_{\mathbb{R}^{p-1}}g(u+\n{w}'\n{w})\,{\rm{d}}\n{w},\quad u\geq0.
\end{equation}

In Theorem \ref{teo-quantis} we show that all the quantiles of the univariate marginal distributions of  
Box--Cox elliptical random vectors are proportional to the respective component of $\n{\mu}$.

\begin{theorem}\label{teo-quantis}
Let $\n{\mu}\in\mathbb{R}_+^p$, $\n{\lambda}\in\mathbb{R}^p$, $\n{\Sigma}(p\times p)>0$ and 
$\n{Y}\sim\mb_p(\n{\mu},\n{\lambda},\n{\Sigma};g)$. The $\alpha$-quantile $y_{k,\alpha}$ of 
$Y_k$, $\alpha\in(0,1)$, $k=1,\ldots,p$, satisfies
\begin{equation}\label{quantismbce}
y_{k,\alpha} =
 \begin{cases}
  \mu_k (1+\lambda_k\sqrt{\sigma_{kk}} s_{k,\alpha})^{1/\lambda_k}, & \text{$\lambda_k\neq 0$},\vspace{0.2cm}\\
  \quad\,\, \mu_k\exp(\sqrt{\sigma_{kk}}s_{k,\alpha}), & \text{$\lambda_k=0$},
 \end{cases}
\end{equation} 
with 
\begin{equation}\label{quantis1}
s_{k,\alpha} =
 \begin{cases}
  F_{U_k}^{-1}(\alpha + (1-\alpha)F_{U_k}(-1/{\lambda_k\sqrt{\sigma_{kk}}})), & \text{$\lambda_k>0$},\vspace{0.2cm}\\
  \,F_{U_k}^{-1}((1+\alpha)F_{U_k}(-1/{\lambda_k\sqrt{\sigma_{kk}}}) - 1), & \text{$\lambda_k<0$},\vspace{0.2cm}\\
  \quad\quad\quad\quad\quad\quad F_{U_k}^{-1}(\alpha), & \text{$\lambda_k = 0$},
 \end{cases}
\end{equation}
where $F_{U_k}$ is the CDF given in (\ref{vble-aux-quan-fda}).
\end{theorem}
\begin{proof}
See Appendix \ref{proof-teo-quantis}. 
\end{proof}

In Theorem \ref{teo-quantis} we stated that, if $\n{Y}\sim\mb_p(\n{\mu},\n{\lambda},\n{\Sigma};g)$, all the
quantiles of $Y_k$, $k=1,\ldots,p$, particularly the median, are proportional to $\mu_k$. This feature of the class of 
Box--Cox elliptical distributions makes it attractive for regression modeling purposes. In Corollary \ref{teo-quantis-logel} 
we establish conditions under which the quantiles of $Y_k$ can be calculated from quantiles of standard symmetric distributions.

\begin{corollary}\label{teo-quantis-logel}
Let $\n{\mu}\in\mathbb{R}_+^p$, $\n{\lambda}\in\mathbb{R}^p$, $\n{\Sigma}(p\times p)>0$ and 
$\n{Y}\sim\mb_p(\n{\mu},\n{\lambda},\n{\Sigma};g)$.  If $\n{\lambda}=\n{0}$ (\textit{i.e.} $\n{Y}\sim\lel_p(\n{\mu},\n{\Sigma};g)$) 
or $\lambda_j\sqrt{\sigma_{jj}}\to0$, $j=1,\ldots,p$, then the $\alpha$-quantile $y_{k,\alpha}$ of $Y_k$, 
$\alpha\in(0,1)$, $k=1,\ldots,p$, is given by $y_{k,\alpha}=\mu_k\exp(\sqrt{\sigma_{kk}}q_{\alpha})$, where 
$q_{\alpha}$ is the $\alpha$-quantile of a standard symmetric distribution with DGF given by (\ref{aux-interp-bcel}).
\end{corollary}

\begin{proof}
Let $\n{\lambda}=\n{0}$ or $\lambda_j\sqrt{\sigma_{jj}}\to0$, $j=1,\ldots,p$ in Theorem \ref{teo-quantis}. From (\ref{quantismbce}) 
and (\ref{quantis1}) it follows that $y_{k,\alpha}=\mu_k\exp(\sqrt{\sigma_{kk}}q_{\alpha})$, where 
$q_{\alpha}=F_{U}^{-1}(\alpha)$, with $U\sim\e_1(0,1;\tilde{g})$, with $\tilde{g}$ being a DGF given by (\ref{aux-interp-bcel}). 
This fact follows because $R(\n{\Delta}_{-k,-k}\n{\lambda}_{-k})=\mathbb{R}^{p-1}$ when 
$\n{\lambda}=\n{0}$, or $R(\n{\Delta}_{-k,-k}\n{\lambda}_{-k}) \to \mathbb{R}^{p-1}$ when $\lambda_j\sqrt{\sigma_{jj}}\to0$, 
$j=1,\ldots,p$.
\end{proof}

Let $\n{Y}\sim\mb_p(\n{\mu},\n{\lambda},\n{\Sigma};g)$. A coefficient of variation based on quantiles for $Y_k$, $k=1,\ldots,p$, is defined as
\begin{equation*}
{\mbox{CV}_{Y_k}}=\frac{3}{4}\frac{(y_{k,3/4} - y_{k,1/4})}{y_{k,1/2}} 
\end{equation*}
(Rigby and Stasinopoulos \cite{RIG-STAS3}) .
Corollary \ref{teo-quantis-logel} allows interpretation of the parameters $\mu_k$ and $\sigma_{kk}$ from their 
relations with quantiles of $Y_k$, $k=1,\ldots,p$. In fact, if $\n{\lambda}=\n{0}$ 
(\textit{i.e.} $\n{Y}\sim\lel(\n{\mu},\n{\Sigma};g)$), then 
$\mu_k=y_{k,1/2}$ \,and\, ${\mbox{CV}_{Y_k}}=1.5\sinh(\sqrt{\sigma_{kk}}q_{3/4})$, where $q_{3/4}$ is the third 
quartile of a standard symmetric distribution with DGF given in (\ref{aux-interp-bcel}). Also, if 
$\n{\lambda}\approx\n{0}$ or $\lambda_j\sqrt{\sigma_{jj}}\approx0$, $j=1,\ldots,p$, then 
$\mu_k\approx y_{k,1/2}$ \,and\, ${\mbox{CV}_{Y_k}}\approx1.5\sinh(\sqrt{\sigma_{kk}}q_{3/4})$. Hence, in 
these cases,  $\mu_k$ is equal or approximately equal to the median of $Y_k$. Moreover, ${\mbox{CV}_{Y_k}}$ depends 
on $\sigma_{kk}$ through the hyperbolic sine function, which is a monotonically increasing function. Therefore, 
$\sigma_{kk}$ can be seen as a relative dispersion parameter of the distribution of $Y_k$.

\section{Parameter estimation}
\label{S:5}

Let $\n{y}_1,\ldots,\n{y}_n$ be the observed values of a random sample $\n{Y}_1,\ldots,\n{Y}_n$ of a random vector
$\n{Y}\sim\mb_p(\n{\mu},\n{\lambda},\n{\Sigma};g)$, with $\n{Y}_i=(Y_{i1},\ldots,Y_{ip})'$, $i=1,\ldots,n$. Let 
$\n{\eta}=(\eta_1,\ldots,\eta_q)'$ be the vector of extra parameters induced by the DGF $g$. The maximum likelihood
estimators of $\n{\mu}$, $\n{\lambda}$, $\n{\Sigma}$ \,and \,$\n{\eta}$, denoted by 
$\n{\widehat{\mu}}$, \,$\n{\widehat{\lambda}}$, $\n{\widehat{\Sigma}}$ \,and \,$\n{\widehat{\eta}}$, respectively, 
will be such that maximize the log-likelihood function $\ell=\sum_{i=1}^{n}\ell_i$, with 
\begin{equation}\label{loglik}
\ell_i = -\log\biggl\{\int_{R(\n{\lambda})}g(\n{w}'\n{\Sigma}^{-1}\n{w})\,{\rm d}\n{w}\biggr\} 
+\log\{g(\n{w}_i'\n{\Sigma}^{-1}\n{w}_i)\} + \sum_{k=1}^{p}(\lambda_k-1)\log y_{ik} - 
\sum_{k=1}^{p}\lambda_k\log\mu_k, 
\end{equation}
where $\n{w}_i=T_{\n{\lambda},\n{\mu}}(\n{y}_i)$. There is no closed form for the maximum likelihood estimators 
$\widehat{\n{\mu}}$, \,$\widehat{\n{\lambda}}$, $\widehat{\n{\Sigma}}$ \,and \,$\widehat{\n{\eta}}$, but they
can be computed using numerical optimization algorithms implemented in computer packages. The number of parameters to be 
estimated is $p(p+5)/2+q$.

Let $\n{\mu}^{(0)}$, $\n{\lambda}^{(0)}$, $\n{\Sigma}^{(0)}$ \,and \,$\n{\eta}^{(0)}$ be the initial values for 
the estimation of $\n{\mu}$, $\n{\lambda}$, $\n{\Sigma}$ \,and \,$\n{\eta}$, respectively. For the choice of 
$\mu_k^{(0)}$, $\lambda_k^{(0)}$ \,and \,$\sigma_{kk}^{(0)}$, $k=1,\ldots,p$, we suggest the estimates obtained by fitting
a Box--Cox symmetric distribution to the $k$-th component of $\n{Y}$, \textit{i.e.} the estimated parameters of 
$Y_k\sim\bcs(\mu_k,\sqrt{\sigma_{kk}},\lambda_k;g)$. As initial values for $\sigma_{jk}$, we suggest 
$\sigma_{jk}^{(0)}=0$, $j\neq k$. Initial values for the extra parameters (if any), $\eta_j^{(0)}$, $j=1,\ldots,q$, 
will depend on the family of distributions considered. For instance, for the multivariate Box--Cox $t$ distribution
we propose as initial value for the degrees of freedom parameter, $\tau^{(0)}$, the corresponding estimate obtained by fitting
a multivariate $t$ distribution to the vector $\n{X}=T_{\n{\lambda}^{(0)},\n{\mu}^{(0)}}(\n{Y})$. 

The main difficulty in implementing an optimization scheme is due to the need of an efficient computation of the 
integral $\int_{R(\n{\lambda})}g(\n{w}'\n{\Sigma}^{-1}\n{w})\,{\rm d}\n{w}$, that appears in (\ref{loglik}). 
This integral depends on the complexity and structure of the DGF $g$ and is computed 
over $R(\n{\lambda})$. Hence, the vector of the extra parameters $\n{\eta}$, the matrix $\n{\Sigma}$ and the vector 
$\n{\lambda}$ are involved in the estimation procedure through this integral. Genz and Bretz \cite{GENZ1} propose 
algorithms to efficiently compute this type of integral over rectangles when $g$ is the DGF of the multivariate normal and
$t$ families. In the class of the log-elliptical distributions ($\n{\lambda}=\n{0}$) the integral 
disappears making the estimation process much easier. 
In this case, the logarithm of the likelihood function is given by $\ell=\sum_{i=1}^{n}\ell_i$, where 
$\ell_i$, $i=1,\ldots,n$, is
\begin{equation*}
\ell_i = -\frac{1}{2}\log(\det(\n{\Sigma}))
+\log\{g(\n{w}_i'\n{\Sigma}^{-1}\n{w}_i)\} - \sum_{k=1}^{p}\log y_{ik},
\end{equation*}
with $\n{w}_i=T_{\n{0},\n{\mu}}(\n{y}_i)$. Here, the unknown quantities to be estimated are
$\n{\mu}$, $\n{\Sigma}$ \,and \,$\n{\eta}$, \textit{i.e.} $p(p+3)/2+q$ parameters.

To evaluate the proposed estimation procedure we conducted simulations with bivariate log-normal, log-$t$, 
Box--Cox normal and Box--Cox $t$ distributions, different sample sizes, namely $n=125, 250, 500$, and $N=5000$ 
Monte Carlo replicates. The random samples of $\n{Y}\sim\mb_p(\n{\mu},\n{\lambda},\n{\Sigma};g)$ were generated 
using Algorithm \ref{algo-bc-elip}.

\begin{algorithm}\label{algo-bc-elip}
\begin{enumerate}
\item[\phantom{alejo}] 
\item Generate a random sample of size $n$, say $\n{w}_1,\ldots,\n{w}_n$, of $\n{W}\sim\te_p(\n{0},\n{\Sigma};R(\n{\lambda});g)$
using Algorithm \ref{algo-elip-trun}.
\item Compute $\n{y}_1=T_{\n{\lambda},\n{\mu}}^{-1}(\n{w}_1),\ldots,\n{y}_n=T_{\n{\lambda},\n{\mu}}^{-1}(\n{w}_n)$.
From Definition \ref{def-bcelip}, $\n{y}_1,\ldots,\n{y}_n$ is a random sample of $\n{Y}\sim\mb_p(\n{\mu},\n{\lambda},\n{\Sigma};g)$.
\end{enumerate}
\end{algorithm}

In each simulation experiment we used the Broyden, Fletcher, Goldfarb, and Shanno (BFGS) optimization algorithm to maximize the 
log-likelihood function with the initial values proposed above. The integral in (\ref{loglik}) was efficiently evaluated using
algorithms proposed by Genz and Bretz \cite{GENZ1}. All the computations were conducted in the \texttt{R} software \cite{R1}.

Let $\widehat{\theta}_1,\ldots,\widehat{\theta}_N$ be the ordered estimated values of a scalar parameter, say $\theta$, 
in $N$ Monte Carlo simulated samples. Let ${\mbox{M}}(\widehat{\theta})$ be the median of 
$\{\widehat{\theta}_1,\ldots,\widehat{\theta}_N\}$. The median bias, denoted by ${\mbox{MB}}(\widehat{\theta})$, 
is given by ${\mbox{MB}}(\widehat{\theta})={\mbox{M}}(\widehat{\theta})-\theta$. The median absolute deviation,  
denoted by ${\mbox{MAD}}(\widehat{\theta})$, is defined as the median of 
$\{|\widehat{\theta}_1-{\mbox{M}}(\widehat{\theta})|,\ldots,|\widehat{\theta}_N-{\mbox{M}}(\widehat{\theta})|\}$. 
Also, let ${\mbox{IQR}}(\widehat{\theta})$ be the interquartile range of $\{\widehat{\theta}_1,\ldots,\widehat{\theta}_N\}$. 
These summaries of the estimates were computed for each simulation experiment and reported in Table \ref{estudo-simul}. 
The figures in this table suggest a suitable behavior of the estimation procedure, because the median biases are close to zero and the 
median absolute deviations and interquartile ranges get smaller as $n$ grows.

\begin{table}[H]
\captionsetup{font=scriptsize}
\scriptsize
\centering
\caption{Median bias (MB), median absolute deviation (MAD) and interquartile range (IQR) of the parameter estimators.}
\label{estudo-simul}
\begin{tabular}{clrrrrrrrrrrrrr}\cmidrule(){1-15}
      &              & \multicolumn{5}{c}{Bivariate log-normal} & \multicolumn{8}{c}{Bivariate Box--Cox $t$} \\ 
\cmidrule(lr){3-7} \cmidrule(lr){8-15}
\multirow{2}{*}{$n$} &         & \multicolumn{1}{c}{$\mu_1$} & \multicolumn{1}{c}{$\mu_2$} & \multicolumn{1}{c}{$\sigma_{11}$} & \multicolumn{1}{c}{$\sigma_{12}$} & \multicolumn{1}{c}{$\sigma_{22}$} & \multicolumn{1}{c}{$\mu_1$} & \multicolumn{1}{c}{$\mu_2$} & \multicolumn{1}{c}{$\lambda_1$} & \multicolumn{1}{c}{$\lambda_2$} & \multicolumn{1}{c}{$\sigma_{11}$} & \multicolumn{1}{c}{$\sigma_{12}$} & \multicolumn{1}{c}{$\sigma_{22}$} & \multicolumn{1}{c}{$\tau$} \\        
      &              & \multicolumn{1}{c}{$8$}   & \multicolumn{1}{c}{$8$} & \multicolumn{1}{c}{$0.8$} & \multicolumn{1}{c}{$-0.5$} & \multicolumn{1}{c}{$1$} & \multicolumn{1}{c}{$20$}  & \multicolumn{1}{c}{$15$}    & \multicolumn{1}{c}{$0.4$} & \multicolumn{1}{c}{$0.3$} & \multicolumn{1}{c}{$0.4$} & \multicolumn{1}{c}{$0.1$}  & \multicolumn{1}{c}{$0.3$} & \multicolumn{1}{c}{$6$} \\ \cmidrule(){1-15}
      & MB           & $-0.02$ & $0.00$ & $-0.02$ & $0.02$  & $-0.03$ & $0.08$ & $-0.01$ & $-0.04$ & $-0.01$ & $0.06$ & $0.06$ & $0.05$ & $2.89$  \\
$125$ & MAD          & $0.59$  & $0.66$ & $0.07$  & $0.07$  & $0.09$  & $0.72$ & $0.42$  & $0.25$  & $0.13$  & $0.31$ & $0.16$ & $0.16$ & $4.62$  \\ \vspace{0.06cm}
      & IQR          & $1.19$  & $1.33$ & $0.15$  & $0.14$  & $0.18$  & $2.96$ & $0.84$  & $0.51$  & $0.27$  & $2.80$ & $0.61$ & $0.36$ & $18.42$ \\ \\
      & MB           & $-0.01$ & $0.00$ & $-0.01$ & $0.01$  & $-0.01$ & $0.05$ & $0.00$  & $-0.02$ & $-0.01$ & $0.04$ & $0.03$ & $0.02$ & $1.48$  \\
$250$ & MAD          & $0.42$  & $0.47$ & $0.05$  & $0.05$  & $0.07$  & $0.50$ & $0.29$  & $0.19$  & $0.10$  & $0.24$ & $0.10$ & $0.11$ & $2.70$  \\ \vspace{0.06cm}
      & IQR          & $0.85$  & $0.93$ & $0.10$  & $0.10$  & $0.13$  & $1.46$ & $0.58$  & $0.38$  & $0.20$  & $0.96$ & $0.28$ & $0.24$ & $6.60$  \\ \\
      & MB           & $0.00$  & $0.00$ & $0.00$  & $0.00$  & $0.00$  & $0.07$ & $0.01$  & $-0.02$ & $-0.01$ & $0.05$ & $0.02$ & $0.02$ & $1.07$  \\ 
$500$ & MAD          & $0.30$  & $0.33$ & $0.04$  & $0.04$  & $0.05$  & $0.40$ & $0.20$  & $0.13$  & $0.07$  & $0.20$ & $0.07$ & $0.08$ & $1.67$  \\ \vspace{0.06cm}
      & IQR          & $0.60$  & $0.66$ & $0.08$  & $0.07$  & $0.09$  & $0.94$ & $0.40$  & $0.27$  & $0.15$  & $0.56$ & $0.17$ & $0.16$ & $3.74$  \\ \cmidrule(){1-15}
      &              & \multicolumn{6}{c}{Bivariate log-$t$} & \multicolumn{7}{c}{Bivariate Box--Cox normal} \\ \cmidrule(lr){3-8} \cmidrule(lr){9-15}
\multirow{2}{*}{$n$} &         & \multicolumn{1}{c}{$\mu_1$} & \multicolumn{1}{c}{$\mu_2$} & \multicolumn{1}{c}{$\sigma_{11}$} & \multicolumn{1}{c}{$\sigma_{12}$} & \multicolumn{1}{c}{$\sigma_{22}$} & \multicolumn{1}{c}{$\tau$} & \multicolumn{1}{c}{$\mu_1$} & \multicolumn{1}{c}{$\mu_2$} & \multicolumn{1}{c}{$\lambda_1$} & \multicolumn{1}{c}{$\lambda_2$} & \multicolumn{1}{c}{$\sigma_{11}$} & \multicolumn{1}{c}{$\sigma_{12}$} & \multicolumn{1}{c}{$\sigma_{22}$} \\        
      &              & \multicolumn{1}{c}{$7$}   & \multicolumn{1}{c}{$10$} & \multicolumn{1}{c}{$1.2$} & \multicolumn{1}{c}{$0.6$}  & \multicolumn{1}{c}{$1.4$} & \multicolumn{1}{c}{$5$} & \multicolumn{1}{c}{$5$}   & \multicolumn{1}{c}{$4$}     & \multicolumn{1}{c}{$-1$}    & \multicolumn{1}{c}{$0.5$} & \multicolumn{1}{c}{$0.6$} & \multicolumn{1}{c}{$0.2$}  & \multicolumn{1}{c}{$0.8$}  \\ \cmidrule(){1-15}
      & MB           & $0.01$  & $0.02$  & $0.00$ & $-0.01$ & $-0.01$ & $0.25$ & $-0.17$ & $-0.06$ & $0.01$ & $0.00$ & $0.00$ & $0.00$ & $-0.01$ \\
$125$ & MAD          & $0.97$  & $1.45$  & $0.22$ & $0.17$  & $0.26$  & $1.76$ & $0.91$  & $0.58$  & $0.08$ & $0.09$ & $0.04$ & $0.02$ & $0.03$  \\ \vspace{0.06cm}
      & IQR          & $1.31$  & $1.98$  & $0.30$ & $0.22$  & $0.35$  & $2.63$ & $1.80$  & $1.15$  & $0.15$ & $0.17$ & $0.08$ & $0.05$ & $0.05$  \\ \\
      & MB           & $-0.01$ & $-0.01$ & $0.00$ & $0.00$  & $0.00$  & $0.15$ & $-0.07$ & $-0.04$ & $0.00$ & $0.00$ & $0.00$ & $0.00$ & $0.00$  \\
$250$ & MAD          & $0.69$  & $1.05$  & $0.16$ & $0.12$  & $0.19$  & $1.18$ & $0.63$  & $0.41$  & $0.05$ & $0.06$ & $0.03$ & $0.02$ & $0.02$  \\ \vspace{0.06cm}
      & IQR          & $0.93$  & $1.43$  & $0.22$ & $0.16$  & $0.25$  & $1.66$ & $1.26$  & $0.82$  & $0.10$ & $0.12$ & $0.05$ & $0.03$ & $0.04$  \\ \\
      & MB           & $0.00$  & $0.00$  & $0.00$ & $0.00$  & $0.00$  & $0.07$ & $-0.02$ & $-0.01$ & $0.00$ & $0.00$ & $0.00$ & $0.00$ & $0.00$  \\
$500$ & MAD          & $0.46$  & $0.72$  & $0.12$ & $0.08$  & $0.13$  & $0.82$ & $0.45$  & $0.29$  & $0.04$ & $0.04$ & $0.02$ & $0.01$ & $0.01$  \\ \vspace{0.06cm}
      & IQR          & $0.62$  & $0.97$  & $0.16$ & $0.11$  & $0.18$  & $1.13$ & $0.90$  & $0.58$  & $0.07$ & $0.09$ & $0.04$ & $0.02$ & $0.03$  \\ \cmidrule(){1-15}
\end{tabular}
\end{table}

\section{Application}
\label{S:6}


The dataset refers to observations of vitamins B2 (in mg), B3 (in mg), B12 (in mcg) and D (in mcg) intakes based on the first 24-h dietary recall 
interview for $n = 136$ older men. The bagplots (Rousseeuw et al. \cite{ROUSS1}) shown in Figure \ref{bagplots-nutric} indicate that the vitamin 
intakes are positively correlated, their bivariate distributions are skewed, and that outliers are present. 

\afterpage{
\begin{figure}[H]
\captionsetup{font=scriptsize}
\centering
\includegraphics[scale=0.64,trim={0.45cm 0.7cm 0cm 0.3cm},clip]{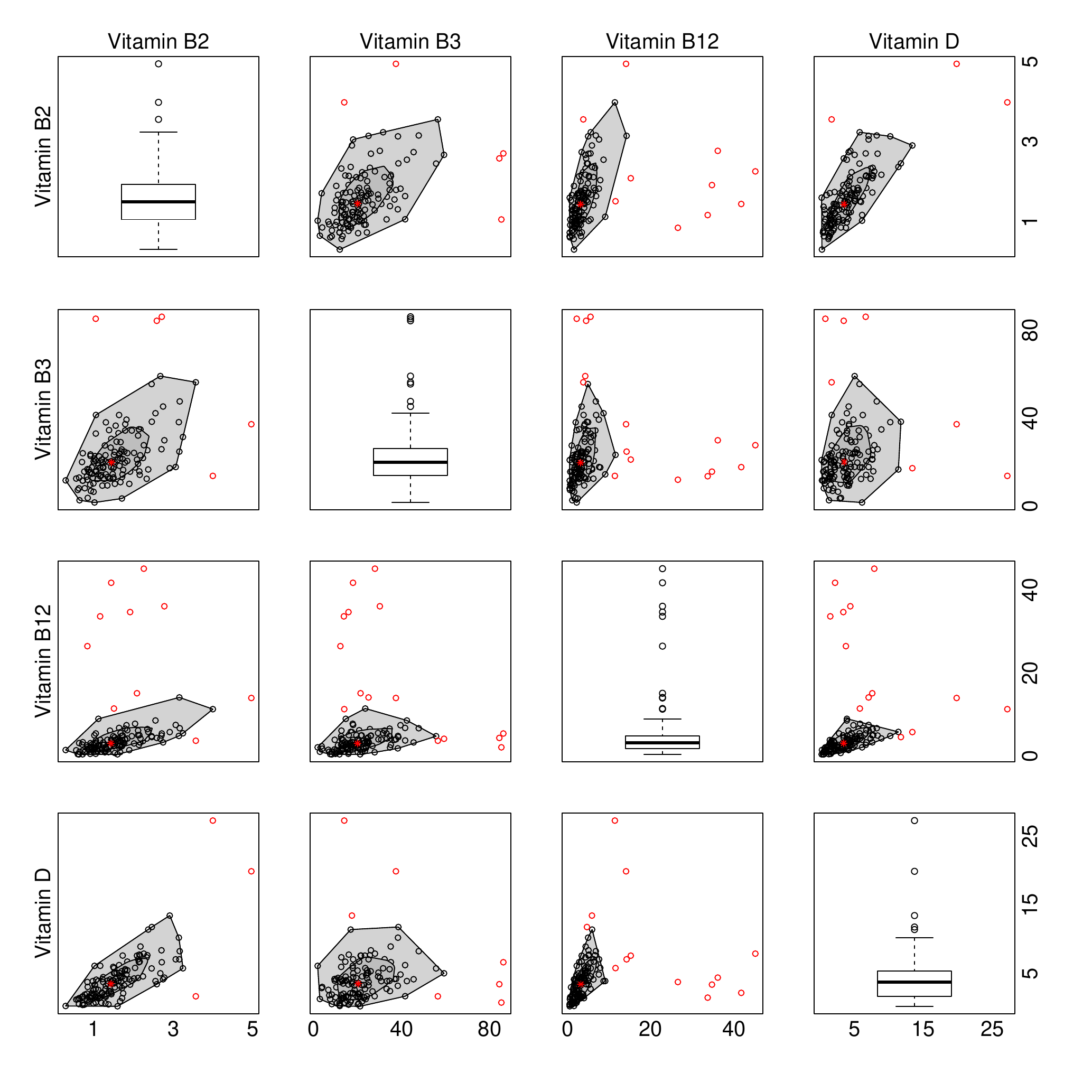}
\caption{Bagplot matrix; nutritional data.}
\label{bagplots-nutric}
\end{figure}
}
For each pair of variables, we fitted bivariate log-normal, log-$t$, Box--Cox normal and Box--Cox $t$ distributions,
and the respective marginal independent distributions; we denote these distributions by $\mln_2$, $\mlt_2$, 
$\mbn_2$, $\mbt_2$, ${\mbox{Ind-}}\mln_1$, ${\mbox{Ind-}}\mlt_1$, ${\mbox{Ind-}}\mbn_1$ and ${\mbox{Ind-}}\mbt_1$, respectively. 
Table \ref{ajustes-nutric} shows the Akaike information criterion (AIC) for each fit. The figures in this table indicate
that the bivariate distributions provide better fit when compared with the respective 
marginal independent distributions. This is not surprising since there is evidence of association among the variables.
Additionally, Table \ref{ajustes-nutric} indicates that the bivariate Box--Cox $t$ distribution gives the best fit 
for the pairs of variables: vitamins B2-D, B3-D and B12-D. Also, the bivariate log-$t$ distribution provides the best fit
for the pairs: vitamins B2-B3, B2-B12 and B3-B12. Hence, the bivariate distributions based on the $t$ distribution provide 
better fit than those based on the normal distribution. This fact is due to the presence of extreme outliers 
(Figure \ref{bagplots-nutric}). 

\afterpage{
\begin{table}[H]
\captionsetup{font=scriptsize}
\scriptsize
\centering
\caption{AIC for the fitting distributions; nutritional data.}
\label{ajustes-nutric}
\begin{tabular}{lrrrrrrrr}\cmidrule(){1-9}
\multicolumn{1}{l}{Variables} & \multicolumn{1}{c}{${\mbox{Ind-}}\mln_1$} & \multicolumn{1}{c}{$\mln_2$} & \multicolumn{1}{c}{${\mbox{Ind-}}\mlt_1$} & \multicolumn{1}{c}{$\mlt_2$} & \multicolumn{1}{c}{${\mbox{Ind-}}\mbn_1$} & \multicolumn{1}{c}{$\mbn_2$} & \multicolumn{1}{c}{${\mbox{Ind-}}\mbt_1$} & \multicolumn{1}{c}{$\mbt_2$} \\ \cmidrule(){1-9} \vspace{0.06cm}
B2-B3  & $1329.87$ & $1291.88$ & $1320.10$ & $\underline{1279.44}$ & $1328.90$ & $1289.57$ & $1323.15$ & $1281.66$             \\ \vspace{0.06cm}
B2-B12 & $946.51$  & $900.99$  & $941.11$  & $\underline{883.24}$  & $945.13$  & $895.20$  & $944.11$  & $885.68$              \\ \vspace{0.06cm}
B2-D   & $916.83$  & $814.27$  & $919.47$  & $802.72$              & $914.99$  & $803.79$  & $918.08$  & $\underline{796.84}$  \\ \vspace{0.06cm}
B3-B12 & $1726.19$ & $1710.38$ & $1707.90$ & $\underline{1689.32}$ & $1720.67$ & $1702.24$ & $1710.27$ & $1690.79$             \\ \vspace{0.06cm}
B3-D   & $1696.51$ & $1689.65$ & $1686.26$ & $1679.20$             & $1690.52$ & $1683.20$ & $1684.75$ & $\underline{1676.79}$ \\ 
B12-D  & $1313.15$ & $1250.09$ & $1307.27$ & $1226.64$             & $1306.76$ & $1233.10$ & $1305.23$ & $\underline{1223.88}$ \\ \cmidrule(){1-9}
\end{tabular}
\end{table}
}

Table \ref{param-ajuste-nutric} gives the estimates (and standard errors) of the parameters of the best fitting model as indicated in Table \ref{ajustes-nutric}. It is noteworthy that the estimated degrees of freedom parameter varies from 
4 to 8, indicating that heavier-than-normal distributions are better suited for fitting the data.
\afterpage{
\begin{table}[H]
\captionsetup{font=scriptsize}
\scriptsize
\centering
\caption{Estimates of the parameters (and standard errors) of the best fitting distribution; nutritional data.}
\label{param-ajuste-nutric}
\begin{tabular}{lrrrrrrrr}\cmidrule(){1-9}
\multicolumn{1}{l}{Variables}        & \multicolumn{1}{c}{$\widehat{\mu}_1$} & \multicolumn{1}{c}{$\widehat{\mu}_2$} & \multicolumn{1}{c}{$\widehat{\lambda}_1$} & \multicolumn{1}{c}{$\widehat{\lambda}_2$} & \multicolumn{1}{c}{$\widehat{\sigma}_{11}$} & \multicolumn{1}{c}{$\widehat{\sigma}_{12}$} & \multicolumn{1}{c}{$\widehat{\sigma}_{22}$} &  \multicolumn{1}{c}{$\widehat{\tau}$} \\ \cmidrule(){1-9} \vspace{0.06cm}
B2-B3 ($\mlt_2$) & $1.45$ $(0.05)$  & $19.91$ $(0.90)$ & \multicolumn{1}{c}{--} & \multicolumn{1}{c}{--} & $0.16$ $(0.02)$ & $0.10$ $(0.02)$ & $0.23$ $(0.04)$ & $6.22$ $(2.18)$ \\ \vspace{0.06cm}
B2-B12 ($\mlt_2$)& $1.46$ $(0.06)$  & $3.10$ $(0.20)$  & \multicolumn{1}{c}{--} & \multicolumn{1}{c}{--} & $0.15$ $(0.02)$ & $0.16$ $(0.03)$ & $0.43$ $(0.08)$ & $4.57$ $(1.36)$ \\ \vspace{0.06cm}
B2-D   ($\mbt_2$)& $1.45$ $(0.06)$  & $3.42$ $(0.23)$  & $0.19$ $(0.17)$        & $0.31$ $(0.11)$        & $0.16$ $(0.03)$ & $0.22$ $(0.04)$ & $0.48$ $(0.08)$ & $7.96$ $(3.57)$ \\ \vspace{0.06cm}
B3-B12 ($\mlt_2$)& $20.10$ $(0.91)$ & $3.13$ $(0.20)$  & \multicolumn{1}{c}{--} & \multicolumn{1}{c}{--} & $0.20$ $(0.04)$ & $0.13$ $(0.03)$ & $0.42$ $(0.08)$ & $3.96$ $(1.12)$ \\ \vspace{0.06cm}
B3-D   ($\mbt_2$)& $19.86$ $(0.94)$ & $3.30$ $(0.23)$  & $0.15$ $(0.14)$        & $0.24$ $(0.12)$        & $0.24$ $(0.04)$ & $0.12$ $(0.03)$ & $0.51$ $(0.08)$ & $7.42$ $(3.03)$ \\
B12-D  ($\mbt_2$)& $3.10$ $(0.20)$  & $3.42$ $(0.24)$  & $-0.19$ $(0.10)$       & $0.15$ $(0.11)$        & $0.45$ $(0.08)$ & $0.31$ $(0.06)$ & $0.47$ $(0.08)$ & $5.50$ $(1.99)$ \\ \cmidrule(){1-9}
\end{tabular}
\end{table}
}

For the bivariate log-$t$ distribution fitted to the pair of vitamins B2-B3 the estimates of $\mu_1$ and $\mu_2$ are
$\widehat{\mu}_1=1.45$ and $\widehat{\mu}_2=19.91$ and correspond to estimates of the median intake of vitamins B2 
and B3 in the population. These estimates are close to the corresponding sample medians ($1.49$ and $19.99$, respectively). 
The estimates of the relative dispersion parameters are $\widehat{\sigma}_{11}=0.16$ and $\widehat{\sigma}_{22}=0.23$; 
hence the relative dispersion of vitamin B2 is estimated to be smaller than that of vitamin B3. 
For the intake of vitamins B12-D the best fit is achieved by the bivariate Box--Cox $t$ distribution. Note that the 
estimated parameters satisfy $\widehat{\lambda}_1\sqrt{\widehat{\sigma}_{11}}=-0.13$ and 
$\widehat{\lambda}_2\sqrt{\widehat{\sigma}_{22}}=0.10$, that are close to zero. Hence, 
$\widehat{\mu}_1=3.10$ and $\widehat{\mu}_2=3.42$ are expected to be close to the sample median
of vitamins B12 and D intakes respectively, and this is in fact the case (the sample medians are $3.24$ and $3.80$, respectively). 
Since $\widehat{\sigma}_{11}=0.45$ and $\widehat{\sigma}_{11}=0.47$, we have that the relative dispersions of vitamins B12 and D intakes are similar.

Figure \ref{diagrama-ajuste-nutric} shows contour plots of the fitted distributions superimposed to the 
scatter plots of the data, and the corresponding PDFs. The plots suggest a reasonable fit for all the pairs of variables.

\begin{figure}[H]
\captionsetup{font=scriptsize}
\centering
\includegraphics[scale=0.64,trim={0.45cm 0.5cm 0cm 0.3cm},clip]{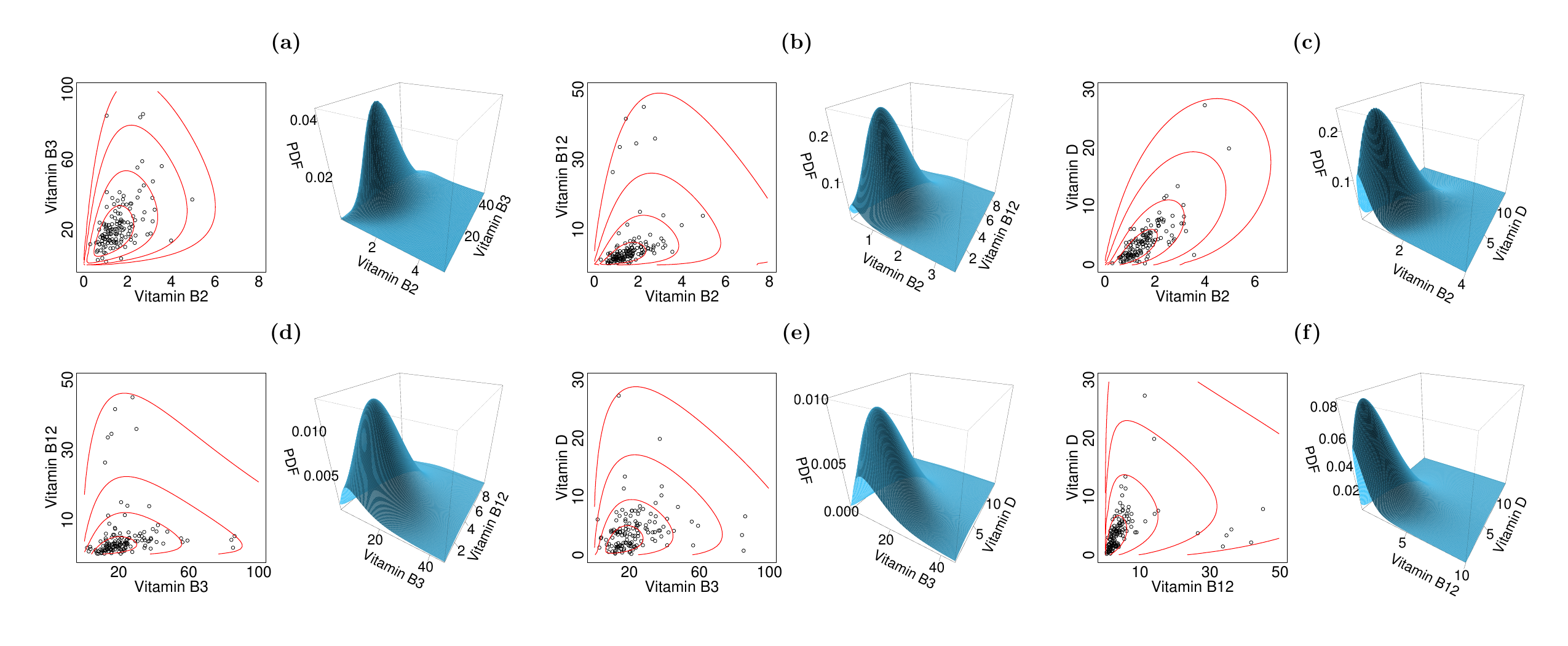}
\caption{Scatter plots overlaid with contour plots and joint PDF of the best fitting distributions; nutritional data.}
\label{diagrama-ajuste-nutric}
\end{figure}

\section{Final remarks}
\label{S:7} 

In this paper we presented a new class of multivariate distributions, the class of Box--Cox elliptical distributions, 
that is suitable for modeling multivariate positive, marginally asymmetric, possibly heavy-tailed data. 
The construction of the Box--Cox distributions uses an extended multivariate Box--Cox 
transformation and the class of truncated elliptical distributions, both defined in this paper. 
We show that the class of Box--Cox elliptical distributions has as special cases the classes of the 
log-elliptical and Box--Cox symmetric distributions. The Box--Cox elliptical distributions allow 
easy parameter interpretation, a desirable feature for modeling purposes. 

Starting from a study of the class of truncated elliptical distributions, we defined and studied the Box--Cox elliptical distributions.
Specifically, we stated useful properties and discussed maximum likelihood estimation issues, generation of random samples, 
interpretation of parameters, and applications.
%

There are some open problems that will be addressed in future papers. The efficiency of the implementation of maximum likelihood estimation
depends on the efficient computatioon of the integral involved in (\ref{boxcoxelipmulti}). The methods proposed by Genz and Bretz \cite{GENZ1} 
to efficiently compute the integral when $g$ is the DGF of the multivariate normal and $t$ distributions allowed us to implement maximum likelihood estimation
for the parameters of the multivariate Box--Cox normal and Box--Cox $t$ distributions. The efficient computation of the integral for other DGFs 
will provide the implementation of maximum likelihood estimation for other distributions in the Box--Cox elliptical class, such as
the multivariate Box--Cox power exponential and Box--Cox slash distributions. Also, extension to regression models is of interest.
The relation of the scale parameters to quantiles of the marginal distributions permits the construction of Box--Cox elliptical
regression models that are able to model the relationship between covariates and quantiles of the response variables. 


\section*{Acknowlegments}

We thank José Eduardo Corrente for providing the data used in this study. 
Funding was provided by Conselho Nacional de Desenvolvimento Científico e Tecnológico -- CNPq (Grant No. 304388-2014-9) and 
Fundação de Amparo à Pesquisa do Estado de São Paulo -- FAPESP (Grant No. 2012/21788-2). The first author received PhD scholarships from 
Coordenação de Aperfeiçoamento de Pessoal de Nível Superior -- CAPES -- and CNPq.

\section*{Appendix}
\appendix
\renewcommand*{\thesection}{\Alph{section}}

\hypertarget{proof-tr-eltrun-condic}{\section{Proof of the Theorem \ref{tr-eltrun-condic}}\label{proof-tr-eltrun-condic}}
The conditional PDF of $W_k|\n{W}_{-k}$, $k=1,\ldots,p$, is given by
\begin{equation*}
f_{W_k|\n{W}_{-k}}(w_k) = \frac{g((\n{w}-\n{\mu})'\n{\Sigma}^{-1}(\n{w}-\n{\mu}))}{\int_{a_k}^{b_k} g((\n{w}-\n{\mu})'\n{\Sigma}^{-1}(\n{w}-\n{\mu}))\,{\rm{d}}w_k},\quad w_k\in (a_k,b_k).
\end{equation*}
From the identity $(\n{w}-\n{\mu})'\n{\Sigma}^{-1}(\n{w}-\n{\mu}) = [(w_k-\mu_{k.-k})/\sigma_{k.-k}]^2+q(\n{w}_{-k})$, 
we get the result.

\hypertarget{proof-carac-bc-elip}{\section{Proof of the Theorem \ref{carac-bc-elip}}\label{proof-carac-bc-elip}}

If $\n{W}=T_{\n{\lambda},\n{\mu}}(\n{Y})\sim{\mbox{TE}l}_p(\n{\xi},\n{\Sigma};R(\n{\lambda});g)$, then 
its PDF is given by
\begin{equation}\label{fdp-w-carac}
f_{\n{W}}(\n{w})=\dfrac{g((\n{w}-\n{\xi})'\n{\Sigma}^{-1}(\n{w}-\n{\xi}))}{\int_{R(\n{\lambda})} g((\n{w}-\n{\xi})'\n{\Sigma}^{-1}(\n{w}-\n{\xi}))\,{\rm d}\n{w}},\quad\n{w}\in R(\n{\lambda}).
\end{equation}
Let $V:R(\n{\lambda})\to R(\n{\lambda})$ be the transformation defined as 
$V(\n{w})=\n{D}_{\n{\alpha}}^{-1}(\n{w}-\n{\xi})$, and let $\n{U}=V(\n{W})$, with Jacobian $J(\n{w}\to\n{u})=\prod_{k=1}^{p}(1+\lambda_k \xi_k)$. The PDF of  $\n{U}$ is
\begin{align*}
f_{\n{U}}(\n{u}) &= \dfrac{g(\n{u}'(\n{D}_{\n{\alpha}}^{-1}\n{\Sigma}\n{D}_{\n{\alpha}}^{-1})^{-1}\n{u})}{\int_{R(\n{\lambda})} g(\n{u}'(\n{D}_{\n{\alpha}}^{-1}\n{\Sigma}\n{D}_{\n{\alpha}}^{-1})^{-1}\n{u})\,{\rm d}\n{u}},\quad\n{u}\in R(\n{\lambda}).
\end{align*}
Hence, $\n{U}\sim{\mbox{TE}l}_p(\n{0},\n{D}_{\n{\alpha}}^{-1}\n{\Sigma}\n{D}_{\n{\alpha}}^{-1};R(\n{\lambda});g)$. Because 
$\n{U}=V(T_{\n{\lambda},\n{\mu}}(\n{Y}))=T_{\n{\lambda},\n{\delta}}(\n{Y})$, where 
$\n{\delta}=T_{\n{\lambda},\n{\mu}}^{-1}(\n{\xi})$, then from Definition \ref{def-bcelip} we have 
$\n{Y}=T_{\n{\lambda},\n{\delta}}^{-1}(\n{U})\sim\mb_p(\n{\delta},\n{\lambda},\n{D}_{\n{\alpha}}^{-1}\n{\Sigma}\n{D}_{\n{\alpha}}^{-1};g)$.

On the other hand, if $\n{Y}\sim\mb_p(\n{\delta},\n{\lambda},\n{D}_{\n{\alpha}}^{-1}\n{\Sigma}\n{D}_{\n{\alpha}}^{-1};g)$, then its PDF is
\begin{equation}\label{fdp-y-carac}
f_{\n{Y}}(\n{y}) = \dfrac{g(\n{w}'(\n{D}_{\n{\alpha}}^{-1}\n{\Sigma}\n{D}_{\n{\alpha}}^{-1})^{-1}\n{w})\prod_{k=1}^{p}\frac{y_k^{\lambda_k-1}}{{\delta_k}^{\lambda_k}}}{\int_{R(\n{\lambda})}g(\n{w}'(\n{D}_{\n{\alpha}}^{-1}\n{\Sigma}\n{D}_{\n{\alpha}}^{-1})^{-1}\n{w})\,{\rm d}\n{w}},\quad\n{w}=T_{\n{\lambda},\n{\delta}}(\n{y}),\quad\n{y}\in\mathbb{R}_{+}^p.
\end{equation}
Now, from the transformation $\n{W}=T_{\n{\lambda},\n{\mu}}(\n{Y})$, with Jacobian 
$J(\n{y}\to\n{w})=\prod_{k=1}^{p}\mu_k(1+\lambda_k w_k)^{1/\lambda_k-1}$, in the PDF (\ref{fdp-y-carac}) we arrive at
PDF (\ref{fdp-w-carac}).

\hypertarget{proof-prop-mbce}{\section{Proof of Theorem \ref{prop-mbce}}\label{proof-prop-mbce}}

\begin{enumerate}
\item From $\n{T}=\n{D_\alpha Y}$, with Jacobian $J(\n{y}\to\n{t})=\prod_{k=1}^{p}a_k^{-1}$, in (\ref{boxcoxelipmulti}), 
we get the PDF of $\n{T}$ as
\begin{equation*}
f_{\n{T}}(\n{t}) = \dfrac{g(\n{w}'\n{\Sigma}^{-1}\n{w})\prod_{k=1}^{p}\frac{t_k^{\lambda_k-1}}{(\alpha_k\mu_k)^{\lambda_k}}}{\int_{R(\n{\lambda})}g(\n{w}'\n{\Sigma}^{-1}\n{w})\,{\rm d}\n{w}},\quad\n{t}\in\mathbb{R}_{+}^p,
\end{equation*}
where $\n{w}=T_{\n{\lambda},\n{\mu}}(\n{D}_{\n{\alpha}}^{-1}\n{t})=T_{\n{\lambda},\n{D_\alpha \mu}}(\n{t})$. Hence,
$\n{T}=\n{D_{\alpha}Y}\sim\mb_p(\n{D_\alpha\mu},\n{\lambda},\n{\Sigma};g)$.

\item Note that the PDF of $\n{Y}$, given in (\ref{boxcoxelipmulti}), can be expressed as
\begin{equation*}
f_{\n{Y}}(\n{y})=\dfrac{g(\n{v}'(\n{D_\beta\Sigma D_\beta})^{-1}\n{v})\prod_{k=1}^{p}\frac{|\beta_k|y_k^{\lambda_k-1}}{\mu_k^{\lambda_k}}}{\int_{R(\n{D}_{\n{\beta}}^{-1}\n{\lambda})}g(\n{v}'(\n{D_\beta\Sigma D_\beta})^{-1}\n{v})\,{\rm d}\n{v}},\quad\n{y}\in\mathbb{R}_{+}^p,
\end{equation*}
where $\n{v}=\n{D_\beta}T_{\n{\lambda},\n{\mu}}(\n{y})$ has its $k$-th component given by
\begin{equation*}\label{vetorv1}
v_k=
 \begin{cases}
  \dfrac{[(y_k/\mu_k)^{\beta_k}]^{\lambda_k/\beta_k}-1}{\lambda_k/\beta_k}, & \text{$\lambda_k\neq0$},\vspace{0.2cm}\\
  \quad\,\,\,\,\log(y_k/\mu_k)^{\beta_k}, & \text{$\lambda_k=0$},
 \end{cases} 
\end{equation*} 
for $k=1,\ldots,p$. From $U_k=(Y_k/\mu_k)^{\beta_k}$, $k=1,\ldots,p$, with Jacobian 
$J(\n{y}\to\n{u})=\prod_{k=1}^{p}\mu_k \beta_k^{-1}u_k^{1/\beta_k-1}$, we arrive at the desired result.

\item Plugging $\n{\lambda}=\n{1}$ in (\ref{boxcoxelipmulti}) we have that the PDF of $\n{Y}$ is
\begin{equation*}
f_{\n{Y}}(\n{y})=\dfrac{g((\n{y}-\n{\mu})'(\n{D_{\mu}\Sigma D_{\mu}})^{-1}(\n{y}-\n{\mu}))\prod_{k=1}^{p}\frac{1}{\mu_k}}{\int_{R(\n{1})}g(\n{w}'\n{\Sigma}^{-1}\n{w})\,{\rm d}\n{w}},\quad\n{y}\in\mathbb{R}_{+}^p.
\end{equation*}
From the change of variables  $\n{w}=\n{D}_{\n{\mu}}^{-1}(\n{y}-\n{\mu})$ we arrive at the desired result.
\end{enumerate}

\hypertarget{proof-bc-marg}{\section{Proof of Theorem \ref{bc-marg}}\label{proof-bc-marg}}

Plugging $\n{\Sigma}_{12}=\n{0}$ in (\ref{margy1}), and then making the change of variables 
$\n{s}=T(\n{w}_2)=\n{\Sigma}_{22}^{-1/2}\n{w}_2$, the marginal PDF of $\n{Y}_1$ is
\begin{align*}
f_{\n{Y}_1}(\n{y}_1)&=\dfrac{\bigr\{\int_{R(\n{\lambda_2})}g(\n{w}_1'\n{\Sigma}_{11}^{-1}\n{w}_1+\n{w}_2'\n{\Sigma}_{22}^{-1}\n{w}_2)\,{\rm d}\n{w}_2\bigl\}\prod_{k=1}^{r}\frac{y_k^{\lambda_k-1}}{\mu_k^{\lambda_k}}}{\int_{R(\n{\lambda}_1)}\bigl\{\int_{R(\n{\lambda}_2)}g(\n{w}_1'\n{\Sigma}_{11}^{-1}\n{w}_1+\n{w}_2'\n{\Sigma}_{22}^{-1}\n{w}_2)\,{\rm d}\n{w}_2\bigl\}\,{\rm d}\n{w}_1}\\
&=\dfrac{g_1(\n{w}_1'\n{\Sigma}_{11}^{-1}\n{w}_1)\prod_{k=1}^{r}\frac{y_k^{\lambda_k-1}}{\mu_k^{\lambda_k}}}{\int_{R(\n{\lambda}_1)}g_1(\n{w}_1'\n{\Sigma}_{11}^{-1}\n{w}_1)\,{\rm d}\n{w}_1},\quad \n{w}_1=T_{\n{\lambda}_1,\n{\mu}_1}(\n{y}_1),\quad\n{y}_1\in\mathbb{R}_{+}^r.
\end{align*}
Note that 
$g_{1}(u) = \int_{T(R(\n{\lambda_2}))}g(u + \n{s}'\n{s})\,{\rm d}\n{s} \leq \int_{\mathbb{R}^{p-r}}g(u+\n{s}'\n{s})\,{\rm{d}}\n{s}=h_1(u)$, 
$u\geq0$, where $h_1$ is such that $\int_{0}^{\infty}t^{r-1}h_1(t^2)\,\textrm{d}t<\infty$ 
(Fang \textit{et al}. \cite[Sec. 2.2]{FANG1}). This completes the proof.

\hypertarget{proof-logell-marg}{\section{Proof of Theorem \ref{logell-marg}}\label{proof-logell-marg}}
Because $\n{Y}\sim\lel_p(\n{\mu},\n{\Sigma};g)$, with 
$g(u)\propto\int_{0}^{\infty}t^{p/2}\exp(-ut/2)\,{\rm d}H(t)$, $u\geq0$, then 
$\n{X}=T_{\n{0},\n{\mu}}(\n{Y})\sim\e_p(\n{0},\n{\Sigma};g)$. Thus, 
$\n{X}_1=T_{\n{0},\n{\mu}_1}(\n{Y}_1)\sim\e_p(\n{0},\n{\Sigma}_{11};g)$ (Kano \cite{KANO1}). Hence, 
$\n{Y}_1\sim\lel_r(\n{\mu}_1,\n{\Sigma}_{11};g)$.

\hypertarget{proof-bc-cond}{\section{Proof of Theorem \ref{bc-cond}}\label{proof-bc-cond}}
The conditional PDF of $\n{Y}_1|\n{Y}_2$ is given by
\begin{equation}\label{cond-proof}
f_{\n{Y}_1|\n{Y}_2}(\n{y}_1)=\frac{g(\n{w}'\n{\Sigma}^{-1}\n{w})\prod_{k=1}^{r}\frac{y_k^{\lambda_k-1}}{\mu_k^{\lambda_k}}}{\int_{R(\n{\lambda}_1)}g(\n{w}'\n{\Sigma}^{-1}\n{w})\,{\rm{d}}\n{\n{w}_1}},\quad\n{w}_1=T_{\n{\lambda}_1,\n{\mu}_1}(\n{y}_1),\quad\n{y}_1\in\mathbb{R}_{+}^r.
\end{equation}
Because $\n{w}'\n{\Sigma}^{-1}\n{w}=\n{u}_1'\bigl(\n{D}_{\n{\alpha}(\n{w}_2)}^{-1}\n{\Sigma}_{11\cdot2}\n{D}_{\n{\alpha}(\n{w}_2)}^{-1}\bigr)^{-1}\n{u}_1+q(\n{w}_2)$, 
where $\n{u}_1=\n{D}_{\n{\alpha}(\n{w}_2)}^{-1}(\n{w}_1-\n{\mu}_1(\n{w}_2))=T_{\n{\lambda}_1,\n{\delta}_1}(\n{y}_1)$, (\ref{cond-proof}) can be expressed as
\begin{equation*}
f_{\n{Y}_1|\n{Y}_2}(\n{y}_1) = \frac{g_{q(\n{w}_2)}\bigl(\n{u}_1'\bigl(\n{D}_{\n{\alpha}(\n{w}_2)}^{-1}\n{\Sigma}_{11\cdot2}\n{D}_{\n{\alpha}(\n{w}_2)}^{-1}\bigr)^{-1}\n{u}_1\bigr)\prod_{k=1}^{r}\frac{y_k^{\lambda_k-1}}{\mu_k^{\lambda_k}(1+\lambda_k\mu_{1k}(\n{w}_2))}}{\int_{R(\n{\lambda}_1)}g_{q(\n{w}_2)}\bigl(\n{u}_1'\bigl(\n{D}_{\n{\alpha}(\n{w}_2)}^{-1}\n{\Sigma}_{11\cdot2}\n{D}_{\n{\alpha}(\n{w}_2)}^{-1}\bigr)^{-1}\n{u}_1\bigr)\,{\rm{d}}\n{\n{u}_1}},\quad\n{y}_1\in\mathbb{R}_{+}^p,
\end{equation*}
where $\n{u}_1=T_{\n{\lambda}_1,\n{\delta}_1}(\n{y}_1)$. Since 
$\prod_{k=1}^{r}\mu_k^{\lambda_k}(1+\lambda_k\mu_{1k}(\n{w}_2)) = \prod_{k=1}^{r}{\delta_k}^{\lambda_k}$ the proof is complete.

\hypertarget{proof-carac-bc-elip-norm}{\section{Proof of Theorem \ref{carac-bc-elip-norm}}\label{proof-carac-bc-elip-norm}}
$\n{Y}_1$ and $\n{Y}_2$ are independent if, and only if, the PDF of
$\n{Y}\sim\mb_p(\n{\mu},\n{\lambda},\n{\Sigma};g)$ given in (\ref{boxcoxelipmulti}) is such that
$f_{\n{Y}}(\n{y})=f_{\n{Y}_1}(\n{y}_1)f_{\n{Y}_2}(\n{y}_2)$.  This condition is satisfied if, and only if, 
$\n{\Sigma}_{12}=\n{0}$ and the DGF $g$ satisfies the functional equation $g(u+v)=g(u)g(v)$, with $u\geq0$ \,and
\,$v\geq0$, for which $g(u)=\exp(-ku)$, for some $k\geq0$, is a solution (Gupta \textit{et al}. \cite[Sec. 1.3]{GUP1}). From 
$\int_{0}^{\infty}t^{p-1}\exp(-kt^2)\,\textrm{d}t=2^{p/2-1}\Gamma(p/2)$, we find that $k=1/2$. Hence, $\n{Y}_1$ and $\n{Y}_2$ are 
independent if, and only if, $\n{\Sigma}_{12}=\n{0}$ and $\n{Y}\sim\mbn_p(\n{\mu},\n{\lambda},\n{\Sigma})$.

\hypertarget{proof-bc-momentos}{\section{Proof of Theorem \ref{bc-momentos}}\label{proof-bc-momentos}}
From (\ref{boxcoxelipmulti}) we have
\begin{equation*}
{\mbox{E}}\biggl(\prod_{k=1}^{p}Y_k^{h_k}\biggr) = \dfrac{\int_{\mathbb{R}_{+}^p}
g(\n{w}'\n{\Sigma}^{-1}\n{w})
\prod_{k=1}^{p}\frac{y_k^{\lambda_k+h_k-1}}{\mu_k^{\lambda_k}}\,{\rm d}\n{y}}{\int_{R(\n{\lambda})}g(\n{w}'\n{\Sigma}^{-1}\n{w})\,{\rm d}\n{w}},
\end{equation*}
where $\n{w}=T_{\n{\lambda},\n{\mu}}(\n{y})$. By making the change of variables $\n{u} =\n{D}_{\n{\mu}}^{-1}\n{y}$ we arrive at the desired result.

\hypertarget{proof-margy_k}{\section{Marginal PDF of $Y_k$}\label{proof-margy_k}}

The function $g_{{\n{\Upsilon}_k}}$ given in (\ref{fdp_marg_aux}) can be defined in $\mathbb{R}$. Hence, we can define a random 
variable $U_k\in\mathbb{R}$ from the PDF
\begin{equation*}
f_{U_k}(u_k)=c_k g_{{\n{\Upsilon}_k}}(u_k),\quad u_k\in\mathbb{R},  
\end{equation*}
where $c_k^{-1} = \int_{-\infty}^{\infty}g_{{\n{\Upsilon}_k}}(t)\,{\rm d}t$. The CDF of $U_k$ is given by 
(\ref{vble-aux-quan-fda}). We now define $S_k\in I(\lambda_k\sqrt{\sigma_{kk}})$ as a random variable $U_k$ truncated on 
$I(\lambda_k\sqrt{\sigma_{kk}})$. The PDF of $S_k$ is given by
\begin{equation}\label{fdp-vable-aux-quant}
f_{S_k}(s_k) = \frac{g_{{\n{\Upsilon}_k}}(s_k)}{\int_{I(\lambda_k\sqrt{\sigma_{kk}})}g_{{\n{\Upsilon}_k}}(s_k)\,{\rm d}s_k},\quad s_k\in I(\lambda_k\sqrt{\sigma_{kk}}).
\end{equation}
From the transformation $S_k=\sigma_{kk}^{-1/2}T_{\lambda_k,\mu_k}(Y_k)$, with Jacobian 
$J(s_k\to y_k)=\sigma_{kk}^{-1/2}\mu_{k}^{-\lambda_k}y_k^{\lambda_k-1}$, we arrive at the PDF of $Y_k$ given in (\ref{margy_k}).

\hypertarget{proof-teo-quantis}{\section{Proof of Theorem \ref{teo-quantis}}\label{proof-teo-quantis}}

Because $\n{Y}\sim\mb_p(\n{\mu},\n{\lambda},\n{\Sigma};g)$, the PDF of $Y_k$, $k=1,\ldots,p$, is given by 
(\ref{margy_k}), where $S_k=\sigma_{kk}^{-1/2}T_{\lambda_k,\mu_k}(Y_k)$ has PDF given in (\ref{fdp-vable-aux-quant}). The
CDF of $S_k$ is
\begin{equation}\label{fda_marg_construc}
F_{S_k}(s_k) =
 \begin{cases}
  \,\,\,\,\dfrac{F_{U_k}(s_k)-F_{U_k}(-1/{\lambda_k\sqrt{\sigma_{kk}})}}{1-F_{U_k}(-1/{\lambda_k\sqrt{\sigma_{kk}})}}, & \text{$\lambda_k > 0$},\vspace{0.2cm}\\
  \dfrac{1 - F_{U_k}(-1/{\lambda_k\sqrt{\sigma_{kk}})}+F_{U_k}(s_k)}{F_{U_k}(-1/{\lambda_k\sqrt{\sigma_{kk}})}}, & \text{$\lambda_k<0$},\vspace{0.2cm}\\
  \quad\quad\quad\quad\,\, F_{U_k}(s_k), & \text{$\lambda_k = 0$},\vspace{0.2cm}
 \end{cases}
\end{equation}
from which we have that the $\alpha$-quantil $y_{k,\alpha}$ of $Y_k$, $\alpha\in(0,1)$, is such that ${\mbox{P}}(Y_k\leq y_{k,\alpha})=\alpha$, or equivalently ${\mbox{P}}[S_k \leq \sigma_{kk}^{-1/2}T_{\lambda_k,\mu_k}(y_{k,\alpha})]=\alpha$. Hence,
\begin{equation*}
y_{k,\alpha} =
 \begin{cases}
  \mu_k (1+\lambda_k\sqrt{\sigma_{kk}} s_{k,\alpha})^{1/\lambda_k}, & \text{$\lambda_k\neq 0$},\vspace{0.2cm}\\
  \quad\,\,\mu_k\exp(\sqrt{\sigma_{kk}}s_{k,\alpha}), & \text{$\lambda_k=0$},
 \end{cases} 
\end{equation*}
where $s_{k,\alpha}$ is such that $F_{S_k}(s_{k,\alpha})=\alpha$, with $F_{S_k}$ given in (\ref{fda_marg_construc}). Therefore, 
$s_{k,\alpha}$ is given by
\begin{equation*}
s_{k,\alpha} =
 \begin{cases}
  F_{U_k}^{-1}(\alpha + (1-\alpha)F_{U_k}(-1/{\lambda_k\sqrt{\sigma_{kk}}})), & \text{$\lambda_k>0$},\vspace{0.2cm}\\
  F_{U_k}^{-1}((1+\alpha)F_{U_k}(-1/{\lambda_k\sqrt{\sigma_{kk}}}) - 1), & \text{$\lambda_k<0$},\vspace{0.2cm}\\
  \quad\quad\quad\quad\quad\quad F_{U_k}^{-1}(\alpha), & \text{$\lambda_k = 0$},
 \end{cases}
\end{equation*}
where $F_{U_k}$ is the CDF given in (\ref{vble-aux-quan-fda}).



\end{document}